%% file: main.tex
\documentclass[pra,onecolumn,notitlepage,nofootinbib,floatfix,11pt,tightenlines]{revtex4-2}

\usepackage[utf8]{inputenc}
\usepackage{datetime}
\usepackage{graphics, graphicx, url, color, physics, cancel, multirow, bbm, soul, mathtools, mathrsfs, amsfonts, amssymb, amsmath, amsthm, dsfont, tabularx, bm, verbatim, tikz, caption, subcaption, xcolor,hyperref, nicefrac,mathdots}
\usepackage{bbm}
\usepackage{bm}
\usepackage{quantikz}
\usepackage[ruled,vlined]{algorithm2e}
\usepackage[boxsize=0.5em, centertableaux]{ytableau}
\usepackage{fullpage}

\def\Pl{{\mathcal{P}_\lambda}}
\def\Pgd{{\mathcal{P}_\gamma^d}}

\def\Qld{{\mathcal{Q}^d_\lambda}}
\def\Qgd{{\mathcal{Q}^d_\gamma}}
\def\Qgdp{{\mathcal{Q}^d_{\gamma,p_\gamma}}}

\def\pg{{{p}_\gamma}}

\def\Cdm{{(\mathbb{C}^d)^{\otimes m}}}

\def\Cdmn{{\big(\mathbbm{C}^{d}\big)^{\otimes m}\otimes \big(\overline{\mathbbm{C}^{d}}\big)^{\otimes n}}}

\newcommand{\RNum}[1]{\uppercase\expandafter{\romannumeral #1\relax}}

\theoremstyle{plain}

\newtheorem{theorem}{Theorem}
\newtheorem{corollary}[theorem]{Corollary}
\newtheorem{lemma}[theorem]{Lemma}
\newtheorem{claim}{Claim}
\newtheorem{fact}[claim]{Fact}

\theoremstyle{definition}
\newtheorem{definition}{Definition}
\newtheorem{example}{Example}

\theoremstyle{remark}

\newtheorem{remark}{Remark}

\def \id {\mathbbm{I}}

\begin{document}
\title{A memory and gate efficient algorithm for unitary mixed Schur sampling}

\author{Enrique Cervero-Mart\'{i}n}
\email{enrique.cervero@u.nus.edu}
\affiliation{Centre for Quantum Technologies, National University of Singapore}

\author{Laura Mančinska}
\email{mancinska@math.ku.dk}
\affiliation{Centre for the Mathematics of Quantum Theory, University of Copenhagen}

\author{Elias Theil}
\email{edmt@math.ku.dk}
\affiliation{Centre for the Mathematics of Quantum Theory, University of Copenhagen}

\begin{abstract}
    We formalize the task of unitary Schur sampling ---an extension of weak Schur sampling--- which is the process of measuring the Young label and the unitary group register of an input $m$ qudit state. Intuitively, this task is equivalent to applying the Schur transform, projecting onto the isotypic subspaces of the unitary and symmetric groups indexed by the Young labels, and discarding of the permutation register. As such unitary Schur sampling is the natural task in processes such as quantum state tomography or spectrum estimation. We generalize this task to unitary mixed Schur sampling to account for the recently introduced mixed Schur-Weyl transform. We provide a streaming algorithm which achieves an exponential reduction in the memory complexity and a polynomial reduction in the gate complexity over naïve algorithms for the task of unitary (mixed) Schur sampling. Further, we show that if the input state has limited rank, the gate and memory complexities of our streaming algorithm as well as the algorithms for the full Schur and mixed Schur transforms are further reduced. Our work generalizes and improves on the results in arXiv2309.11947.
\end{abstract}

\maketitle

\tableofcontents

\section{Introduction}
The symmetries arising from the actions of the unitary and permutation groups play a crucial role in the understanding of quantum systems. For systems composed of multiple copies of a quantum state, an identical unitary action on all copies commutes with an arbitrary permutation of the copies.
At the core of the unitary and permutational symmetries lies \emph{Schur-Weyl duality}.
Individually, the action of the unitary/symmetric group decomposes the state space of $n$ qudits into an \emph{isotypic sum} of \emph{minimal invariant subspaces} which are left unchanged under the unitary/symmetric group actions, each occuring with some multiplicity. 
Schur-Weyl duality states that the minimal invariant subspaces of the unitary group coincide precisely with the multiplicity spaces of the symmetric group, and vice-versa. 
Explicitly, Schur-Weyl duality is captured by the relation \cite{GW98}
\begin{align}\label{SWDualityIntro}
    \Cdm \cong \bigoplus_{\lambda\vdash_{d} m} \Pl\otimes\Qld,
\end{align}
where $\Pl$ and $\Qld$ correspond to the minimal invariant subspaces (also known as \emph{irreducible representations}, or \emph{irreps}) of the symmetric and unitary groups, respectively, and the partitions $\lambda$ of $m$ are known as the \emph{Young labels}. Any basis that spans $\Cdm$ and respects the decomposition on the right hand side is known as a \emph{Schur basis}. A \emph{(quantum) Schur transform} is therefore a unitary mapping from the computational basis of $\Cdm$ to a Schur basis, and can be implemented as a quantum circuit following \cite{BCH05,HarrowTh05,KiSt18,Krovi19}.
In the qubit case, there is an implementation with only a logarithmic number of auxiliary registers following \cite{WS23}.
Quantum Schur transform is a fundamental building block in many quantum computation and information protocols that seek to exploit the symmetries of quantum mechanical systems. Examples include algorithms for spectrum estimation and related problems \cite{KeyWer01, ChMi06, DoWr15}, quantum state tomography \cite{DoWr16, Haetal17}, quantum data compression \cite{HaMa03, Hayashi16}, distortion-free entanglement concentration \cite{MaHa07, BlCrGo14}, encoding into decoherence-free subspaces \cite{ZaRa97, KnLaVi00, KeBaLiWha01,BaconThesis03}, superreplication \cite{ChiYa15, ChiYa16}, de-Finetti theorems \cite{KoeMit09, Gross21}, geometric quantum machine learning\cite{Ragoneetal22, Nguyenetal22, Schatzkietal22}, and most recently quantum majority voting \cite{Buhrmanetal22}.

Usually the Schur transform is used as a subroutine to measure the irrep label $\lambda$, and then use the post-measurement state on $\Pl\otimes\Qld$ for further computation. 
This process is typically called \emph{weak Schur sampling}, and it can be naïvely implemented by first applying the Schur transform on $m$ qubits $U_{\rm Sch}^m$ and then applying the PVM $\{\Pi_\lambda = \mathbb{I}_{\Pl}\otimes\mathbb{I}_{\Qld}\}_{\lambda\vdash_{d} m}$. For many applications, however, we are only interested in the post-measurement state on the unitary register $\mathcal{Q}_{\lambda}^{d}$ and the permutation register can be discarded (e.g.  \cite{DoWr15,Haetal17,BlCrGo14,Buhrmanetal22}). This motivates us to introduce the task of \emph{unitary Schur sampling}, which is the first contribution of this paper.
\newline

\noindent
\textbf{I: Unitary Schur Sampling.} We formally introduce the task of \emph{unitary Schur sampling} in section \ref{subs:unitary_schur_sampling}. This is equivalent to first performing weak Schur sampling to obtain a state on $\mathcal{P}_{\lambda}\otimes\mathcal{Q}_{\lambda}^{d}$, and subsequently discarding the permutation register to obtain a state on $\mathcal{Q}_{\lambda}^{d}$.
\newline

For many algorithms using this reduced form of Schur sampling, it is also natural to assume that the inputs are received sequantially in a \emph{streaming manner}. Prime examples for this are spectrum estimation or state tomography, where the input in question consists of i.i.d. states generated by a fixed circuit queried repeatedly. 
We can therefore expect that this streaming setting, together with the fact that unitary Schur sampling does not need the full power of the Schur transform, will allow us to construct more memory- or time-efficient algorithms as compared to existing approaches. Our second contribution is therefore:
\newline

\noindent
\textbf{II: An Algorithm for Unitary Schur Sampling.} We present Algorithm \ref{alg:streamalgo}, which performs unitary Schur sampling in a streaming manner, and which has an \emph{exponential improvement in memory complexity} in the regime of constant $d$ compared to naïve implementations via the Schur transform. The results are summarized in Theorem \ref{thm:memory_gate_complexity} and a comparison with other algorithms is given in Tables \ref{tab:comparison_algorithms_unitary_Schur_sampling} and \ref{tab:comparison_similar_algorithms_unitary_Schur_sampling}.
\newline

In the setting of spectrum estimation and state tomography, rank promises on the input have been widely investigated. If the input state $\rho$ has rank $r\ll d$, then there are a number of results showing that the sample complexity decreases by a factor of $(r/d)$ (e.g. \cite{DoWr16, Haetal17}). This is due to a close interplay between rank and irrep label $\lambda$. It is therefore natural to ask whether similar improvements can be obtained for other algorithms using Schur-Weyl duality. Our third contribution concerns this.
\newline

\noindent
\textbf{III: Rank Promise.}
If every input qudit is restricted to a subspace of $\mathbbm{C}^{d}$ with dimension $r\leq d$, we obtain an \emph{improvement in memory and time complexity} compared to known implementations via the Schur transform. In particular, we improve the scaling by $(r/d)$ and $(r/d)^{3}$ respectively. We obtain this improvement in time complexity also for the implementation of the full Schur transform given in \cite{HarrowTh05}, see Corollary \ref{cor:rank_improvement_Schur_transform}.
\newline

Schur-Weyl duality is generalized by \emph{mixed Schur-Weyl duality} which considers systems of $m+n$ qudits over state space $\Cdmn$, where $\overline{\mathbbm{C}^{d}}$ is the dual space of $\mathbbm{C}^{d}$. This space has the natural action of $SU(d)$ given by $U\mapsto U^{\otimes m}\otimes \overline{U}^{\otimes n}$, where $\overline{U}$ denotes the complex conjugate. Here, the action of the symmetric group is replaced by the action of the so-called \emph{walled Brauer algebra} $\mathcal{B}_{m,n}^d$ ---a generalization of the symmetric group algebra. Mixed Schur-Weyl duality states that under these two actions
\begin{align}\label{mSWDualityIntro}
    \Cdmn \cong \bigoplus_{\gamma\vdash_{d}(m,n)} \Pgd\otimes \Qgd,
\end{align}
where the spaces $\Pgd$ and $\Qgd$ are the irreps of $\mathcal{B}_{m,n}^d$ and $SU(d)$ respectively, and the $\gamma\vdash_{d}(m,n)$ indexing the irreps are known as \emph{staircases}. In particular, Schur-Weyl duality is a special case of mixed Schur-Weyl duality for $n=0$. The mixed Schur transform can be implemented on a quantum computer following \cite{Nguyen_2023,Grinko_2023} and may be used to efficiently implement unitary equivariant channels \cite{GO23}, which have applications in port-based teleportation \cite{IH08,KM+21}, asymetric cloning \cite{NPR21}, covariant quantum error correction \cite{KL22}, entanglement detection \cite{HKMV22} or black-box transformations \cite{YSM23}. It is therefore interesting to develop new algorithms for mixed Schur-Weyl duality. As our fourth contribution, we generalize the previous two contribution to the mixed Schur-Weyl duality setting.
\newline

\noindent
\textbf{IV: Unitary Mixed Schur Sampling.} We generalize the task of unitary Schur sampling and our algorithm to mixed Schur-Weyl duality. In particular, we obtain the same exponential improvement in memory complexity, and $(r/d)$ improvement in time complexity if restricted to input qudits of reduced rank $r\leq d$. This improvement translates to the algorithm for the full mixed Schur transform given in \cite{Grinko_2023,Nguyen_2023}. A comparison with other algorithms is given in Table \ref{tab:comparison_algorithms_mixed_unitary_Schur_sampling}.
\newline

Since mixed Schur-Weyl duality is a generalization of Schur-Weyl duality, we state all results in Theorem \ref{thm:memory_gate_complexity} and Corollary \ref{cor:rank_improvement_Schur_transform} in terms of mixed Schur-Weyl duality. The reduced case of Schur-Weyl duality can be obtained by setting $n=0$.

This paper is structured as follows: In Section \ref{sec:results} we present the task of unitary (mixed) Schur sampling and state the mathematical results. The main part of this paper then starts with an exposition on the mathematical preliminaries of mixed Schur-Weyl duality in Section \ref{sec:prelims}. In Section \ref{sec:steaming_algorithm}, we present our algorithm and in particular, we prove its correctness in Section \ref{sec:alg_proof}. We further discuss its gate and memory complexity in \ref{sec:complexity}, relegating some mathematical details to the Appendices.

\section{Technical Results}
\label{sec:results}
\subsection{Unitary Schur sampling}
\label{subs:unitary_schur_sampling}

In this paper, we introduce the task of \emph{unitary (mixed) Schur sampling}. Our motivation is that many algorithms using Schur-Weyl duality do not need the full Schur transform, but only the state on the unitary irrep $\Qld$. Examples for this are spectrum estimation and related problems \cite{KeyWer01, ChMi06, DoWr15}, state tomography \cite{DoWr16,Haetal17}, entanglement concentration \cite{MaHa07,BlCrGo14}, purification \cite{Cirac_99}, optimal cloning \cite{Werner_1998}, and quantum majority vote \cite{Buhrmanetal22}. In general, any algorithm that uses Schur-Weyl duality and where the output is invariant under permutations of the input has this property. Formally, we have the following definition:

\begin{definition}[Unitary Schur sampling]
    \label{def:unitary_Schur_sampling}
    Let $\rho$ be a state on $\Cdm$. Let further $p(\lambda,\rho)$ be the probability of performing weak Schur sampling on $\rho$ and measuring the irrep label $\lambda\vdash_{d} m$, and let $\rho_{\lambda}$ be the post-measurement state on $\mathcal{P}_{\lambda}\otimes \mathcal{Q}_{\lambda}^{d}$. Then we define \emph{unitary Schur sampling} as the task of taking $\rho$ as input and obtaining $\lambda$ and $\Tr_{\mathcal{P}_{\lambda}}[\rho_{\lambda}]$ with probability $p(\lambda,\rho)$.
\end{definition}

It is important to note here that while unitary Schur sampling is defined via weak Schur sampling and subsequent tracing out of the $\mathcal{P}_{\lambda}$ register, this is not the only possible implementation. In particular, our algorithm described in Section \ref{sec:alg_proof} follows a different method and leads to the improvements in memory complexity described in Theorem \ref{thm:memory_gate_complexity} and Table \ref{tab:comparison_algorithms_unitary_Schur_sampling}.

I the context of the nrecentlyintroduced mixed Schur-Weyl duality \cite{Nguyen_2023,Grinko_2023} of Eq.~\eqref{mSWDualityIntro},  we define\emph{unitary mixed Schur sampling} cin a similar manner

\begin{definition}[Unitary mixed Schur sampling]
    \label{def:mixed_unitary_Schur_sampling}
    Let $\rho$ be a state on $\Cdmn$. Let further $p(\gamma,\rho)$ be the probability of applying the mixed Schur transform on $\rho$ and measuring the irrep label $\gamma\vdash_{d} (m,n)$, and let $\rho_{\gamma}$ be the post-measurement state on $\mathcal{P}_{\gamma}^{d}\otimes \mathcal{Q}_{\gamma}^{d}$. Then we define \emph{unitary mixed Schur sampling} as the task of taking $\rho$ as input and obtaining $\gamma$ and $\Tr_{\mathcal{P}_{\gamma}^{d}}[\rho_{\gamma}]$ with probability $p(\gamma,\rho)$.
\end{definition}

\subsection{Time and memory complexity}
We propose Algorithm \ref{alg:streamalgo} detailed in Section \ref{sec:alg_proof} to perform unitary (mixed) Schur sampling, as introduced in the previous section.

We take as gate complexity the number of elementary ($CNOT$, $H$, $Z^{\nicefrac{1}{4}}$) gates needed to perform the algorithm, and we take as memory complexity the maximum number of qubits needed to store the state at any point during the algorithm. For the gate and memory complexity of our algorithm we obtain the following result, which is a combination of Lemmata \ref{lem:wmSsCorrectness}, \ref{lem:memory_gate_complexity_1} and \ref{lem:memory_gate_complexity_2}.

\begin{theorem}[Complexity of Algorithm \ref{alg:streamalgo}]
    \label{thm:memory_gate_complexity}
    Algorithm \ref{alg:streamalgo} performs the task of unitary mixed Schur sampling up to accuracy $\epsilon$ using a stream of $m+n$ input qudits, $M$ qubits of memory and $T$ elementary ($CNOT$, $H$, $Z^{\nicefrac{1}{4}}$) gates, where
    \begin{align}
        &M=O\big(d^2\log_{2}^{p}(d,m,n,1/\epsilon)\big) \, ,\\
        &T=O\big((m+n)d^4\log_{2}^{p}(d,m,n,1/\epsilon)\big) \, ,
    \end{align}
    with $p\approx 1.44$. Let further $S,S'\subseteq\mathbbm{C}^{d}$ be subspaces with $\dim S=r$ and $\dim S'=r'$. If the input is restricted to the subspace
    \begin{align}\label{rank_restriction}
        S^{\otimes m}\otimes \overline{S'}^{\otimes n}\subseteq \big(\mathbbm{C}^d\big)^{\otimes m}\otimes \big(\overline{\mathbbm{C}^d}\big)^{\otimes n} \, ,
    \end{align}
    then
    \begin{align}
        \label{equ:memory_complexity_reduced}
        &M=O\big((r+r')d\log_{2}^{p}(d,m,n,1/\epsilon)\big) \, ,\\
        \label{equ:gate_complexity_reduced}
        &T=O\big((m+n)(r+r')^3d\log_{2}^{p}(d,m,n,1/\epsilon)\big) \, .
    \end{align}
\end{theorem}

\begin{remark}
    If the gate complexity is not of concern, we can adapt our algorithm to obtain
    \begin{align}
        M=O\big(d(r+r')\log_{2}(n+m)\big) \, .
    \end{align}
    However, this changes the gate complexity to
    \begin{align}
        T=O\big(d(r+r')(m+n)^{2d(r+r')+1}\log^{p}_2(m,n,1/\epsilon)\big) \, .
    \end{align}
    For more details, see Remark \ref{rem:straightforward_implementation_memory} and the precursor work \cite{Cervero_2024}.
\end{remark} 

\begin{remark}
    Since $r,r'\leq d$, the algorithms achieving Eqs. \eqref{equ:memory_complexity_reduced} and \eqref{equ:gate_complexity_reduced} only yield an improvement if $r,r'\ll d$.
\end{remark}

Altogether, Theorem \ref{thm:memory_gate_complexity} represents exponential savings in memory for the regime of constant $d$, and a gate complexity reduction of $(d/(r+r'))^{3}$ for reduced rank input. For a full comparison between our algorithm and other algorithms see Tables \ref{tab:comparison_algorithms_unitary_Schur_sampling}, \ref{tab:comparison_similar_algorithms_unitary_Schur_sampling} and \ref{tab:comparison_algorithms_mixed_unitary_Schur_sampling}.

Table \ref{tab:comparison_algorithms_unitary_Schur_sampling} compares our algorithm for unitary Schur sampling to naïve implementations via the Schur transform. Table \ref{tab:comparison_similar_algorithms_unitary_Schur_sampling} compares our algorithm to generalized phase estimation algorithms proposed in \cite{HarrowTh05}. These algorithms do not perform full unitary Schur sampling, as they don't provide the post-measurement state on $\mathcal{Q}_{\lambda}$. However, they can be used to perform measurements in the irrep label $\lambda$ and the Gelfand-Tsetlin basis (explained in Section \ref{subsubs:GT_Basis}) on $\mathcal{Q}_{\lambda}^{d}$. This is useful for e.g. spectrum estimation, where only $\lambda$ is necessary. Finally, Table \ref{tab:comparison_algorithms_mixed_unitary_Schur_sampling} compares our algorithm for unitary mixed Schur sampling to naïve implementations via the mixed Schur transform.

In the tables, the algorithms marked by $^{*}$ originally had $p\approx 3.97$ from the Solovay-Kitaev Theorem, but using \cite{Kuperberg23} as a primitive instead yields an automatic improvement to $p\approx 1.44$. In addition, we do not have a precise estimation of the memory complexity for the algorithms marked by $\dagger$, however they all require at least $m$ memory-qudits to store the input.

\begin{table*}[htbp!]
\renewcommand{\arraystretch}{1.5}
\begin{center}
    \begin{tabular}{ |p{0.16\linewidth}|p{0.30\linewidth}|p{0.26\linewidth}|p{0.23\linewidth}|} 
    \hline
        Algorithm & Gates & Memory & Notes \\ 
    \hline
        Bacon+ \cite{BCH05} & $O\big(md^4\log_{2}^{p}(d,m,1/\epsilon)\big)$ & $O\big((m+d^2)\log_{2}^{p}(d,m,1/\epsilon)\big)$ & $*$ \\
    \hline
        Krovi \cite{Krovi19} & $O\big(\text{poly}(m,\log(d),\log(1/\epsilon))\big)$ & $\Omega(m)$ & $\dagger$ \\
    \hline    
        Wills+ \cite{WS23} & $O\big(m^3\log(m)\log(m/\epsilon)\big)$ & $O(m)$ & For qubits \\
    \hline
        This Work & $O\big(mr^{3}d\log_{2}^{p}(d,m,1/\epsilon)\big)$ & $O\big(rd\log_{2}^{p}(d,m,1/\epsilon)\big)$ & $r\leq d$ is the reduced input dimension \\
    \hline
\end{tabular}
\end{center}
\caption{Comparison of gate and memory complexity between different algorithms for unitary Schur sampling.
}
\label{tab:comparison_algorithms_unitary_Schur_sampling}
\end{table*}

\begin{table*}[htbp!]
\renewcommand{\arraystretch}{1.5}
\begin{center}
    \begin{tabular}{ |p{0.16\linewidth}|p{0.30\linewidth}|p{0.26\linewidth}|p{0.23\linewidth}|} 
    \hline
        Algorithm & Gates & Memory & Notes \\ 
    \hline
        Harrow (GPE I) \cite{HarrowTh05} &  $O\big((\text{poly}(n)+n\log(d))\log^{p}(1/\epsilon)\big)$ & $\Omega(m)$ & no post-measurement state on $\mathcal{Q}_{\lambda}^{d}$, $\dagger$, $*$ \\
    \hline
        Harrow (GPE II) \cite{HarrowTh05} &  $O\big(d\text{poly}(n,\log(n),\log(1/\epsilon))\big)$ & $\Omega(m)$ & only measurements in Gelfand-Tsetlin basis for $\mathcal{Q}_{\lambda}^{d}$ possible, $\dagger$ \\
    \hline
        This Work & $O\big(mr^{3}d\log_{2}^{p}(d,m,1/\epsilon)\big)$ & $O\big(rd\log_{2}^{p}(d,m,1/\epsilon)\big)$ & $r\leq d$ is the reduced input dimension \\
    \hline
\end{tabular}
\end{center}
\caption{Comparison of gate and memory complexity between our algorithm and generalized phase estimation (GPE).
}
\label{tab:comparison_similar_algorithms_unitary_Schur_sampling}
\end{table*}

\begin{table*}[htbp!]
\renewcommand{\arraystretch}{1.5}
\begin{center}
    \begin{tabular}{ |p{0.15\linewidth}|p{0.35\linewidth}|p{0.30\linewidth}|p{0.15\linewidth}|} 
    \hline
        Algorithm & Gates & Memory & Notes \\ 
    \hline
        Nguyen, Grinko+ \cite{Nguyen_2023,Grinko_2023} & $O\big((n+m)d^4\log_{2}^{p}(d,m,n,1/\epsilon)\big)$ & $O\big((n+m+d^2)\log_{2}^{p}(d,m,n,1/\epsilon)\big)$ & $*$ \\
    \hline
        This Work & $O\big((m+n)(r+r')^{3}d\log_{2}^{p}(d,m,n,1/\epsilon)\big)$ & $O\big((r+r')d\log_{2}^{p}(d,m,n,1/\epsilon)\big)$ & $r,r'\leq d$ are the reduced input dimensions \\
    \hline
\end{tabular}
\end{center}
\caption{Comparison of gate and memory complexity between different algorithms for unitary mixed Schur sampling.
}
\label{tab:comparison_algorithms_mixed_unitary_Schur_sampling}
\end{table*}

Further, noting that the algorithms for the Schur transform and mixed Schur transform of \cite{Bacon_2006,Nguyen_2023,Grinko_2023} are similar to our Algorithm \ref{alg:streamalgo}, we get the following corollary:

\begin{corollary}[Complexity of (mixed) Schur transform]
    \label{cor:rank_improvement_Schur_transform}
    If the input is restricted to the subspaces in Eq.~\eqref{rank_restriction}, the algorithms for the mixed Schur transform described in \cite{Nguyen_2023,Grinko_2023} have gate complexity
    \begin{align}
        T=O\big((m+n)(r+r')^3d\log_{2}^{p}(d,m,n,1/\epsilon)\big) \, ,
    \end{align}
    with $p\approx 1.44$. In particular, setting $n=0$ we obtain that the algorithm for the Schur transform described in \cite{BCH05} has gate complexity
    \begin{align}
        T=O\big(mr^3d\log_{2}^{p}(d,m,1/\epsilon)\big) \, .
    \end{align}
\end{corollary}

This corollary gives a gate complexity reduction of $(d/(r+r'))^{3}$ for the implementations of the (mixed) Schur transform described in \cite{Bacon_2006,Nguyen_2023,Grinko_2023}. This is true for general applications of the (mixed) Schur transform outside of the scope of unitary (mixed) Schur sampling.

\section{Preliminaries}
\label{sec:prelims}

\subsection{Mixed Schur-Weyl duality}

\subsubsection{Schur-Weyl duality}
We assume some familiarity with Schur-Weyl duality and refer to \cite{GW98, HarrowTh05} for an in-depth introduction, as well as \cite{HarrowTh05, BCH05, KiSt18, Krovi19} for implementations of the Schur transform. 
In this section, we briefly introduce these concepts for the sake of establishing notation. 

First, we write $\lambda\vdash_{d}m$ for a partition
\begin{align}
    \lambda\in\mathbbm{N}_{0}^{d} \, , \quad m\geq\lambda_1\geq...\geq\lambda_d\geq 0 \, , \quad \lambda_1+...+\lambda_d=m \, .
\end{align}
Partitions are often described by Young diagrams, which are composed of $d$ rows with $\lambda_{k}$ boxes in the $k$-th row. 

Second, we describe the natural actions of $\mathcal{S}_{m}$ and $SU(d)$ on $\Cdm$ given by
\begin{align}
    Q^d_m(U)\ket{i_1,...,i_m}&:= U^{\otimes m}\ket{i_1,...,i_m}\,,\qquad \textrm{for}\qquad U\in SU(d)\,,\\
    P_m(\sigma)\ket{i_1,...,i_m}&:= 
    \ket{i_{\sigma^{-1}(1)},...,i_{\sigma^{-1}(m)}}\,,\qquad 
    \textrm{for}\qquad \sigma\in \mathcal{S}_m.
\end{align}
The decomposition of these actions into irreps are linked by Schur-Weyl duality:
\begin{theorem}[Schur-Weyl duality]
    The vector space $\big(\mathbbm{C}^d\big)^{\otimes m}$ decomposes under the action of $\mathcal{S}_{m}$ and $SU(d)$ in the following way
    \begin{align}\label{sw_duality_def}
        \big(\mathbbm{C}^d\big)^{\otimes m} \stackrel{\mathcal{S}_{m}\times SU(d)}{\cong} \bigoplus_{\lambda\vdash_{d}m}\mathcal{P}_{\lambda}\otimes\mathcal{Q}_{\lambda}^d \, .
    \end{align}
    We call a unitary mapping that realizes the above isomorphism a \textit{Schur transform}, and denote it $U_{\rm Sch}^{m}$.
\end{theorem}

\subsubsection{Mixed Schur-Weyl duality}
Mixed Schur-Weyl duality is a generalization of Schur-Weyl duality, involving $m$ `input' qudits and $n$ `output' qudits. On the input qudits, unitaries act like before, whereas on the output, they are complex conjugated. On the space $\Cdmn$, this setup describes the representation of $U\in SU(d)$ given by
\begin{align}
    Q_{m,n}^{d}(U):=U^{\otimes m}\otimes\overline{U}^{\otimes n} \, ,
\end{align}
where $\overline{U}$ is the complex conjugate of $U$. As with Schur-Weyl duality, the action of $SU(d)$ via $Q_{m,n}^{d}(\cdot)$ admits a decomposition into irreps $\Qgd$ indexed by \emph{staircases} $\gamma$.
Here, we write $\gamma\vdash_{d} (m,n)$ for the \emph{staircase} of $m,n$ given by
\begin{align}
    \label{equ:staircase_def_1}
    \gamma\in\mathbbm{Z}^{d} \, , \quad \gamma_1\geq\cdots\geq\gamma_d \, , \quad \gamma_1+...+\gamma_d=m-n \, ,
\end{align}
\begin{align}
    \label{equ:staircase_def_2}
    \sum_{\gamma_i>0}\gamma_i\leq m \quad , \quad \sum_{\gamma_i<0}|\gamma_i|\leq n \, .
\end{align}

On the other hand, in the space $\Cdmn$, the action of the symmetric group is replaced by the action of the \emph{walled Brauer algebra} $\mathcal{B}^d_{m,n}$.
The elements of the algebra $\mathcal{B}^d_{m,n}$ are complex combinations of so called \emph{partially transposed permutations} of $m+n$ elements.
A partially transposed permutation $\sigma^{T}\in\mathcal{B}^d_{m,n}$ is obtained from a permutation $\sigma\in\mathcal{S}_{m+n}$ as shown in Figure \ref{fig:wBrauer}, by swapping the last $n$ nodes between the first and second rows of the permutation $\sigma\in\mathcal{S}_{m+n}$.
\begin{figure}[h]
    \centering
    \scalebox{0.8}{\input{wBrauer_tikz}}
    \caption{Element of walled Brauer algebra with $m=2$ and $n=3$ (on the right) corresponding to the partial transpose of the permutation $(132)(45)$ in cyclic notation (on the left).}
    \label{fig:wBrauer}
\end{figure}
It follows that if either $m=0$ or $n=0$, then $\mathcal{B}_{m,n}^d$ is equivalent to the algebra generated by $\mathcal{S}_n$ or $\mathcal{S}_m$ respectively. 

Similar to~\cite{Nguyen_2023}, we define the representation $P^d_{m,n}:\mathcal{B}_{m,n}\rightarrow M(d(m+n))$ acting on $\Cdmn$ as
\begin{align}
    P^d_{m,n}(\sigma)(\ket{i_1,..,i_{m+n}}):=\sum_{j_1,...,j_{m+n}=1}^d\sigma_{j_1,...,j_{m+n}}^{i_1,...,i_{m+n}}\ket{j_1,...,j_{m+n}}\,,
\end{align}
where $\sigma_{j_1,...,j_{m+n}}^{i_1,...,i_{m+n}}=\prod_{(t,b)\in \sigma}\delta_{k_t,k_b}$. Here, the $(t,b)\in \sigma$ are the connections of the diagram representing $\sigma$, and $k_t,k_b$ are the indices corresponding to the node.
In particular, if $n=0$ this equals the action of permutations $\sigma\in\mathcal{S}_m$ on $\Cdm$.

\begin{example}
    For $\sigma$ depicted in Figure \ref{fig:wBrauer}, we get
    \begin{align}
        \sigma_{j_1,...,j_{m+n}}^{i_1,...,i_{m+n}}=\delta_{i_1,i_3}\delta_{i_2,j_1}\delta_{i_4,j_5}\delta_{i_5,j_4}\delta_{j_2,j_3} \, .
    \end{align}
\end{example}

From here on, we focus on the matrix algebra $\mathcal{A}^d_{m,n}$ generated by the matrices $P_{m,n}^d(\sigma)$ for $\sigma\in\mathcal{B}^d_{m,n}$,
as we are only interested in irreps that appear in this algebra. The action of $\mathcal{A}^d_{m,n}$ on $\Cdmn$ admits a decomposition into irreps $\Pgd$ also indexed by \emph{staircases} $\gamma\vdash_{d}(m,n)$.
In contrast to the irreps $\Pl$ of the symmetric group, the irreps $\Pgd$ of $\mathcal{A}^d_{m,n}$ depend on the dimension $d$.
Similarly to Schur-Weyl duality, we have:

\begin{theorem}[Mixed Schur-Weyl duality~\cite{Koike_89, Nguyen_2023,Grinko_2023}]
    \label{thm:mixed_SW_duality}
    The vector space $\Cdmn$ decomposes under the actions of $\mathcal{A}_{m,n}^{d}$ and $SU(d)$ in the following way
    \begin{align}\label{equ:msw_duality_def}
        \Cdmn \stackrel{\mathcal{A}_{m,n}^{d}\times SU(d)}{\cong} \bigoplus_{\gamma\vdash_{d} (m,n)}\mathcal{P}_{\gamma}^{d}\otimes\mathcal{Q}_{\gamma}^d \, .
    \end{align}
    We call a unitary mapping that realizes the above isomorphism a mixed Schur transform, and denote it $U^{m,n}_{\rm mSch}$
\end{theorem}

\begin{remark}
    Setting $n=0$ recovers Schur-Weyl duality, so from now on we will only discuss mixed Schur-Weyl duality, and all the following results will hold as well for Schur-Weyl duality.
\end{remark}

\subsection{Bases for $\mathcal{P}_{\gamma}^{d}$ and $\mathcal{Q}_{\gamma}^{d}$}
There is a canonical choice of basis in both spaces  $\mathcal{P}_{\gamma}^{d}$ and $\mathcal{Q}_{\gamma}^{d}$, which is built upon the idea of finding a sequence of subgroups or subalgebras and label each base vector according to the irreps of these subgroups. 
From now on, we will talk about \textit{the} (mixed) Schur transform $U_{\rm mSch}^{m,n}$ which has this specific choice of basis. For a more detailed discussion, we refer to \cite{Grinko_2023}.

\subsubsection{Young-Yamanouchi basis for $\mathcal{P}_{\gamma}^{d}$}
\label{subsubs:Young-Yamanouchi_Basis}
For a staircase $\gamma\vdash_{d}(m,n)$ and $i\in[d]$ we define
\begin{align}
    \gamma+e_{i}:=(\gamma_{1},...,\gamma_{i}+1,...,\gamma_{d}) \, ,\\
    \gamma-e_{i}:=(\gamma_{1},...,\gamma_{i}-1,...,\gamma_{d}) \, .
\end{align}
With this in hand, we can define for $\gamma\vdash_{d}(m,n)$ the sets
\begin{align}
    \{\gamma+\ydiagram{1}\}&:=\{\nu\vdash_{d}(m+1,n)\,:\, \nu=\gamma+e_{i} \text{ for some } i\in[d]\} \, , \\
    \{\gamma-\ydiagram{1}\}&:=\{\nu\vdash_{d}(m,n+1)\,:\, \nu=\gamma-e_{i} \text{ for some } i\in[d]\} \, .
\end{align}

\begin{example}
    Let us take as example the staircase $\gamma\vdash_{3}(3,2)$ given by $\gamma=(2,1,-2)$. Then we have
    \begin{align}
        \gamma-\ydiagram{1}=\{(1,1,-2), (2,0,-2), (2,1,-3)\} \, .
    \end{align}
    In terms of Young diagrams, we can understand the positive entries of $\gamma$ as boxes to the right, and the negative entries of $\gamma$ as boxes to the left. This gives us
    \begin{align}
        \ytableausetup{boxsize=1em}
        \ydiagram{2+2,2+1,2} \quad - \quad \ydiagram{1}=\left\{ \, \ydiagram{2+1,2+1,2} \quad , \quad \ydiagram{2+2,2+0,2} \quad , \quad \ydiagram{3+2,3+1,3} \, \right\} \, .
        \ytableausetup{boxsize=0.5em}
    \end{align}
    We further remark that for all $\nu\in\gamma-\ydiagram{1}$, we have $\nu\vdash_{3}(3,3)$.
\end{example}

A canonical basis for $\mathcal{P}_{\gamma}^{d}$ is the Young-Yamanouchi basis, which is labelled by sequences $p_{\gamma}=(\gamma^{1},...,\gamma^{m+n})$ of staircases with the properties
\begin{align}
    \label{equ:Young-Yamanouchi_basis}
    \begin{split}
    &\gamma^{1}=(1,0,...,0) \quad \text{ if $m\geq 1$}\, ,\\
    &\gamma^{1}=(0,...,0,-1) \quad \text{ if $m=0$}\, ,\\
    &\gamma^{m+n}=\gamma \, ,\\
    &\gamma^{i+1}\in\{\gamma^{i}+\ydiagram{1}\} \quad \text{ if } i<m \, , \\
    &\gamma^{i+1}\in\{\gamma^{i}-\ydiagram{1}\} \quad \text{ if } i\geq m \, .
    \end{split}
\end{align}

\begin{example}
    For $\gamma\vdash_{3}(3,3)$ given by $\gamma=(1,1,-2)$, one such sequence $p_{\gamma}$ is given by
    \begin{align}
        p_{\gamma}=\left((1,0,0), (1,1,0), (2,1,0), (2,1,-1), (2,1,-2), (1,1,-2)\right)\, .
    \end{align}
    In terms of Young diagrams, this is given by
    \begin{align}
        \ytableausetup{boxsize=1em}
        p_{\gamma}=\left( \,  \ydiagram{1,0,0} \quad , \quad \ydiagram{1,1,0} \quad , \quad \ydiagram{2,1,0} \quad , \quad \ydiagram{1+2,1+1,1} \quad , \quad \ydiagram{2+2,2+1,2} \quad , \quad \ydiagram{2+1,2+1,2} \, \right) \, .
        \ytableausetup{boxsize=0.5em}
    \end{align}
    Notice that after adding all $m=3$ boxes, we subtract $n=3$ boxes. Subtracted boxes can either remove an existing box from a row, or create a ``negative'' box in an empty row.
\end{example}

To understand the properties of this basis, we look at the tower of subalgebras
\begin{align}
    \label{equ:walled_Brauer_algebra_subgroups}
    \mathcal{A}_{1,0}^{d}\subseteq \mathcal{A}_{2,0}^{d} \subseteq ... \mathcal{A}_{m,0}^{d} \subseteq \mathcal{A}_{m,1}^{d} \subseteq ... \subseteq \mathcal{A}_{m,n-1}^{d} \subseteq \mathcal{A}_{m,n}^{d} \,.
\end{align}
These inclusions are similar to viewing $\mathcal{S}_{n-k}\subseteq\mathcal{S}_{n}$ by fixing the last $k$ entries. Under the action of $\mathcal{A}_{m-1,n}^{d}\subseteq\mathcal{A}_{m,n}^{d}$, the space $\mathcal{P}_{\gamma}^{d}$ is not an irrep anymore and instead we have:
\begin{align}
    \label{equ:walled_Brauer_algebra_irreps_subgroup_m}
    \mathcal{P}_{\gamma}^{d}\stackrel{\mathcal{A}_{m-1,n}^{d}}{\cong}\bigoplus_{\substack{\nu\vdash_{d}(m-1,n): \\ \gamma\in\{\nu-\ydiagram{1}\}}}\mathcal{P}_{\nu}^{d} \,.
\end{align}
A basis vector labelled by $p_{\gamma}=(\gamma^{1},...,\nu,\gamma)$ is then part of the $\mathcal{A}_{m-1,n}^{d}$ irrep $\mathcal{P}_{\nu}^{d}$  given in Eq. \eqref{equ:walled_Brauer_algebra_irreps_subgroup_m}. 
We can then continue the decomposition according to \eqref{equ:walled_Brauer_algebra_subgroups} to obtain $p_{\gamma}=(\gamma^{1},...,\gamma^{m+n})$ as described above.

\subsubsection{Gelfand-Tsetlin basis for $\mathcal{Q}_{\gamma}^{d}$}
\label{subsubs:GT_Basis}
Let now $\gamma\vdash_{d}(m,n)$ and $\nu\vdash_{d-1}(m',n')$, with $m'\leq m$ and $n'\leq n$. Then we write $\nu\preceq\gamma$ if
\begin{align}
    \gamma_{1}\geq\nu_{1}\geq\gamma_{2}\geq ... \geq \gamma_{d-1}\geq\nu_{d-1}\geq\gamma_{d} \,.
\end{align}
Similar to the Young-Yamanouchi basis, we want to describe the base vectors of $\mathcal{Q}_{\gamma}^{d}$ by investigating how they transform under the tower of subgroups
\begin{align}
    SU(1)\subseteq SU(2) \subseteq ... \subseteq SU(d-1) \subseteq SU(d) \, .
\end{align}
The so-called \textit{Gelfand-Tsetlin basis} has basis vectors which are also labelled by sequences $q_{\gamma}=(\gamma^{1},...,\gamma^{d})$ with
\begin{align}
    &\gamma^{i}\vdash_{i}(m,n) \quad \text{ for $1\leq i \leq d$ } \, \\
    &\gamma^{i}\preceq\gamma^{i+1} \quad \text{ for $1\leq i \leq d-1$ } \, \\
    &\gamma^{d}=\gamma \, .
\end{align}
Again, we can view $\mathcal{Q}_{\gamma}^{d}$ as a (reducible) representation of $SU(d-1)$ which decomposes into the irreps
\begin{align}
    \label{equ:Gelfand-Tsetlin_decomposition}
    \mathcal{Q}_{\gamma}^{d}\stackrel{SU(d-1)}{\cong}\bigoplus_{\nu\preceq\gamma}\mathcal{Q}_{\nu}^{d-1} \,.
\end{align}
The vector labelled by $q_{\gamma}$ lies in $\mathcal{Q}_{\gamma^{d-1}}^{d-1}$ and is in general part of $\mathcal{Q}_{\gamma^{i}}^{i}$ when looking at the irrep decomposition of $SU(i)$ for $i\leq d$.

\subsection{The mixed Schur transform}
In this section, we describe how the mixed Schur transform is implemented in a quantum computer using \emph{Clebsch-Gordan transforms} and \emph{dual Clebsch-Gordan transforms}, which we also introduce in the following.
In Appendix \ref{app:building_the_CG_transform}, we additionally describe how to implement the (dual) Clebsch-Gordan transform, as the exact algorithm is important for our reduced rank result.

\subsubsection{The Clebsch-Gordan transform}

The algorithms of \cite{BCH05,HarrowTh05,KiSt18} for the Schur transform and of \cite{Nguyen_2023,Grinko_2023} for the mixed Schur transform proceed in an iterative manner by constructing the (mixed) Schur transform of $k+1$ qudits from the (mixed) Schur transform of $k$ qudits.
The building blocks used to do this are the {Clebsch-Gordan transforms}.

An important fact in the representation theory of $SU(d)$ is (see e.g. \cite{GW98})
\begin{align}
\begin{split}\label{eq:cg_dcg}
    \mathcal{Q}_{\nu}^d\otimes \mathbbm{C}^d \stackrel{SU(d)}{\cong} \bigoplus_{\gamma\in\{\nu+\ydiagram{1}\}}\mathcal{Q}_{\gamma}^d \, , \\
    \mathcal{Q}_{\nu}^d\otimes \overline{\mathbbm{C}^d} \stackrel{SU(d)}{\cong} \bigoplus_{\mu\in\{\nu-\ydiagram{1}\}}\mathcal{Q}_{\mu}^d \, .
\end{split}
\end{align}
We call the corresponding unitaries $U_{\rm CG}^{\nu}$ and $U_{\rm dCG}^{\nu}$ the \textit{(dual) Clebsch-Gordan transform}.\footnote{This is a special case of the decompositions mentioned in Lemma \ref{lem:restriction_tensor_product_irreps}.} 
Instead of looking at individual irreps, we can also inspect the Schur-Weyl decomposition of the whole space $\Cdmn$ and define the \textit{super Clebsch-Gordan transform} as
\begin{align}
     U^{m,n}_{\rm CG}:=\bigoplus_{\gamma\vdash_{d} (m,n)}\mathbbm{I}_{\mathcal{P}_{\gamma}^{d}}\otimes U^{\gamma}_{\rm CG} \, , \\
     U^{m,n}_{\rm dCG}:=\bigoplus_{\gamma\vdash_{d} (m,n)}\mathbbm{I}_{\mathcal{P}_{\gamma}^{d}}\otimes U^{\gamma}_{\rm dCG} \, .
\end{align}
It can be shown (up to a reordering of the spaces $\Qgd$ on the right hand side) \cite{BCH05,HarrowTh05,Nguyen_2023,Grinko_2023} that
\begin{align}
    U^{m,n}_{\rm CG}:\left(\bigoplus_{\gamma\vdash_{d} (m,n)}\mathcal{P}_{\gamma}^{d}\otimes\mathcal{Q}_{\gamma}^d\right)\otimes\mathbbm{C}^{d}&\rightarrow\bigoplus_{\gamma\vdash_{d} (m+1,n)}\mathcal{P}_{\gamma}^{d}\otimes\mathcal{Q}_{\gamma}^d\,, \\
    U^{m,n}_{\rm dCG}:\left(\bigoplus_{\gamma\vdash_{d} (m,n)}\mathcal{P}_{\gamma}^{d}\otimes\mathcal{Q}_{\gamma}^d\right)\otimes\overline{\mathbbm{C}^{d}}&\rightarrow\bigoplus_{\gamma\vdash_{d} (m,n+1)}\mathcal{P}_{\gamma}^{d}\otimes\mathcal{Q}_{\gamma}^d\,.
\end{align}

Lastly, for $\nu\in\{\gamma+\ydiagram{1}\}$ and $\mu\in\{\gamma-\ydiagram{1}\}$ we also define the following projection onto the irreps $\mathcal{Q}_{\mu}^{d}$ for $\mu\in\{\gamma+\ydiagram{1}\}$ and $\mathcal{Q}_{\nu}^{d}$ for $\nu\in\{\gamma-\ydiagram{1}\}$ as
\begin{align}\label{eq:mid_projections}
    \Pi^\gamma_\mu: \bigoplus_{\mu'\in\{\gamma+\ydiagram{1}\}} \mathcal{Q}^d_{\mu'}\rightarrow \mathcal{Q}^d_\mu\,,&\qquad \Pi^\gamma_\mu := \bigoplus_{\mu'\in\{\gamma+\ydiagram{1}\}} \delta_{\mu',\mu}\,\mathbbm{I}_{\mathcal{Q}^d_{\mu'}} \, , \\
    \Pi_{\nu}^\gamma: \bigoplus_{\nu'\in\{\gamma-\ydiagram{1}\}} \mathcal{Q}_{\nu'}^{d}\rightarrow \mathcal{Q}_{\nu}^{d}\,,&\qquad \Pi^{\gamma}_{\nu} := \bigoplus_{\nu'\in\{\gamma-\ydiagram{1}\}} \delta_{\nu',\nu}\,\mathbbm{I}_{\mathcal{Q}_{\nu'}^{d}} \, ,\label{eq:mid_projections2}
\end{align}
where $\delta_{\alpha,\beta}$ is the Kronecker symbol. Here we also view the co-domains of the maps $\Pi^\gamma_\mu$ and $\Pi_{\nu}^\gamma$ as the irreps $\mathcal{Q}^d_\mu$ respectively $\mathcal{Q}^d_\nu$ instead of the full space. We emphasize that these maps form a projection valued measure (PVM) on the collection of spaces respectively indexed by staircases in $\{\gamma+\ydiagram{1}\}$ and $\{\gamma-\ydiagram{1}\}$, which we will use in our algorithm for unitary Schur sampling.

\subsubsection{Relation between Schur transform and Clebsch-Gordan transform}
In the following we assume $m\geq1$, and the case for $m=0$ follows from a similar argument. We first remind ourselves that for the staircases $(1,0,...,0)\vdash_{d}(1,0)$ and $(0,...,0,-1)\vdash_{d}(0,1)$ correspond to
\begin{align}
    \mathcal{Q}_{(1,0,...,0)}^{d}\cong \mathbbm{C}^d \quad , \quad \mathcal{Q}_{(0,...,0,-1)}^{d}\cong \overline{\mathbbm{C}^d} \,.
\end{align}
This leads us to the following chain of equivalences, which we obtain by repeatedly using the super Clebsch-Gordan transform and the super dual Clebsch-Gordan transform
\begin{align}
    \label{equ:decomposition_Cxdxn}
    \begin{split}
    \Cdmn
    &\stackrel{SU(d)}{\cong} \mathcal{Q}_{(1,0,...,0)}^{d}\otimes\left(\mathbbm{C}^d\right)^{\otimes m-1}\otimes \left(\overline{\mathbbm{C}^d}\right)^{\otimes n}\\
    &\stackrel{SU(d)}{\cong} \bigoplus_{\gamma^{2}\in\{(1,0,...,0)+\ydiagram{1}\}}\mathcal{Q}_{\gamma^{2}}^{d}\otimes\left(\mathbbm{C}^d\right)^{\otimes m-2}\otimes \left(\overline{\mathbbm{C}^d}\right)^{\otimes n}\\
    &\quad\stackrel{\vdots}{}\\
    &\stackrel{SU(d)}{\cong} \bigoplus_{\gamma^{2}\in\{(1,0,...,0)+\ydiagram{1}\}}...\bigoplus_{\gamma^{m}\in\{\gamma^{m-1}+\ydiagram{1}\}}\mathcal{Q}_{\gamma^{m}}^{d}\otimes \left(\overline{\mathbbm{C}^d}\right)^{\otimes n}\\
    &\quad\stackrel{\vdots}{}\\
    &\stackrel{SU(d)}{\cong} \bigoplus_{\gamma^{2}\in\{(1,0,...,0)+\ydiagram{1}\}}...\bigoplus_{\gamma^{m}\in\{\gamma^{m-1}+\ydiagram{1}\}}\bigoplus_{\gamma^{m+1}\in\{\gamma^{m}-\ydiagram{1}\}}...\bigoplus_{\gamma^{m+n}\in\{\gamma^{m+n-1}-\ydiagram{1}\}}\mathcal{Q}_{\gamma^{m+n}}^{d} \, .
    \end{split}
\end{align}
We see immediately that the sequences $((1,0,...,0),\gamma^{2},...,\gamma^{m+n})$ appearing in each summand of the right hand side of Eq.~\eqref{equ:decomposition_Cxdxn} conform to the same rules as the basis vectors of the Young-Yamanouchi basis described in Eq. \eqref{equ:Young-Yamanouchi_basis}. For each such sequence $p_{\gamma}=(\gamma^{1},...,\gamma^{m+n})$ in Eq. \eqref{equ:decomposition_Cxdxn}, we can then label the corresponding irrep of $SU(d)$ as $\Qgdp$. Since the partition of $\Cdmn$ into irreps is unique, this expression has to agree with the mixed Schur-Weyl duality in Theorem \ref{thm:mixed_SW_duality}. This implies that $\Qgdp$ corresponds to fixing the vector $\ket{p_{\gamma}}\in\mathcal{P}_{\gamma}^{d}$ for a given staircase $\gamma$ and as such, Eq. \eqref{equ:decomposition_Cxdxn} can be rewritten as
\begin{align}
    \label{equ:alternative_mixed_Schur_Weyl_duality}
    \Cdmn \stackrel{\mathcal{A}_{m,n}^{d}\times SU(d)}{\cong} \bigoplus_{\gamma\vdash_{d} (m,n)}\bigoplus_{\pg\in\Pgd}\Qgdp\,.
\end{align}

\subsubsection{An iterative construction of the Schur transform}
For a staircase $\gamma\vdash_{d}(m,n)$ and a path $p_{\gamma}=(\gamma^{1},...,\gamma^{m+n}=\gamma)$, we define
\begin{align}
    \Pi^{m,n}_{\gamma,\pg} &:= \bigoplus_{\nu\vdash_{d}(m,n)}\bigoplus_{p_\nu\in\mathcal{P}^d_\nu}\delta_{\nu,\gamma}\delta_{p_\nu,\pg}\mathbbm{I}_{\mathcal{Q}_{\nu,p_\nu}^d}\label{eq:gpg_projections}\, . \\
    \Pi^{m,n}_\gamma &:= \sum_{\pg\in\Pgd} \Pi^{m,n}_{\gamma,\pg} = \bigoplus_{\nu\vdash_{d}(m,n)}\delta_{\nu,\gamma}\mathbbm{I}_{\mathcal{P}_\nu^d}\otimes\mathbbm{I}_{\mathcal{Q}_{\nu}^d}
    \label{eq:g_projections}
\end{align}
to be the projections from the direct sums in Eqs. \eqref{equ:msw_duality_def} and \eqref{equ:alternative_mixed_Schur_Weyl_duality} respectively.

We now consider projecting onto a given irrep $\mathcal{Q}_{\gamma}^{d}$ labelled by $p_{\gamma}$ in two different ways. Let us take two staircases $\nu\vdash_{d}(m,0)$ and $\gamma\in\{\nu+\ydiagram{1}\}$, together with two paths
\begin{align}
    &p_{\nu}=(\nu^{1},...,\nu^{m}=\nu)\,,\\
    &p_{\gamma}=(\nu^{1},...,\nu^{m},\gamma) \, .
\end{align}
For one, we can perform the decomposition of Eq. \eqref{equ:decomposition_Cxdxn} and then take the projection $\Pi_{\gamma,p_{\gamma}}^{m+1,0}$. On the other hand, we can also perform the decomposition in Eq. \eqref{equ:decomposition_Cxdxn} until the penultimate step and take the projection $\Pi_{\nu,p_{\nu}}^{m,0}$. Then we apply the Clebsch-Gordan transform $U_{\rm CG}^{m,0}$ and finally project onto the $\pg$-th irrep $\gamma$ via $\Pi^\nu_\gamma$ as per Eq. \eqref{eq:mid_projections}. Both ways will lead to the irrep $\mathcal{Q}_{\gamma}^{d}$ labelled by $p_{\gamma}$. Therefore we get
\begin{align}
    \Pi_{\gamma}^{m,0} U^{m,0}_{\rm CG}(\Pi_{\nu,p_{\nu}}^{m,0}U_{\rm Sch}^{m,0}\otimes\mathbbm{I}_{d})=\Pi_{\gamma,p_{\gamma}}^{m,0}U_{\rm Sch}^{m+1,0} \, .
\end{align}
We can argue in a similar way about $\nu\vdash_{d}(m,n)$ and $\gamma\in\{\nu-\ydiagram{1}\}$ to obtain the following proposition.

\begin{lemma}
    \label{lem:CG_Schur_proj_commute}
    Let $\nu\vdash_{d}(m,0)$ and $\gamma\vdash_{d}(m+1,0)$ with $\gamma\in\{\nu+\ydiagram{1}\}$. 
    Further let $p_{\nu}=(\nu^{1},...,\nu)$ be a path of staircases that go to $\nu$, and $p_{\gamma}=(\nu^{1},...,\nu,\gamma)$ be the path of staircases obtained from $p_{\gamma}$ by adding a box to $\nu$. Then we have
    \begin{align}\label{eq:super_CG}
        \Pi_{\gamma}^{m,0} U^{m,0}_{\rm CG}(\Pi_{\nu,p_{\nu}}^{m,0}U_{\rm Sch}^{m,0}\otimes\mathbbm{I}_{d})=\Pi_{\gamma,p_{\gamma}}^{m,0}U_{\rm Sch}^{m+1,0} \, .
    \end{align}
    Let now $\nu\vdash_{d}(m,n)$ and $\gamma\vdash_{d}(m,n+1)$ with $\gamma\in\{\nu-\ydiagram{1}\}$. Further let $p_{\nu}$ and $p_{\gamma}$ be as above. Then we have
    \begin{align}\label{eq:super_dCG}
        \Pi_{\gamma}^{m,n} U^{m,n}_{\rm dCG}(\Pi_{\nu,p_{\nu}}^{m,n}U_{\rm Sch}^{m,n}\otimes\mathbbm{I}_{d})=\Pi_{\gamma,p_{\gamma}}^{m,n}U_{\rm Sch}^{m,n+1} \, .
    \end{align}
\end{lemma}

We remark that summing over all projections for all possible paths and irreps just gives the identity. Therefore we directly obtain this iterative decomposition into (dual) Clebsch-Gordan transforms for the Schur transform given by

\begin{corollary}
    \label{cor:iterative_CG_Schur}
    We have
    \begin{align}
        U_{\rm Sch}^{m+1,0}=U_{\rm CG}^{m,0}(U_{\rm Sch}^{m,0}\otimes\mathbbm{I}_{d}) \, , \\
    U_{\rm Sch}^{m,n+1}=U_{\rm dCG}^{m,n}(U_{\rm Sch}^{m,n}\otimes\mathbbm{I}_{d}) \, .
    \end{align}
\end{corollary}

The insights from Lemma \ref{lem:CG_Schur_proj_commute} and Corollary \ref{cor:iterative_CG_Schur} form the foundation of previous works about efficient quantum (dual) Schur transforms \cite{Bacon_2006,Nguyen_2023,Grinko_2023}, as well as our own contribution, as they give an iterative way to construct the Schur transform on a quantum computer.

\section{Unitary (mixed) Schur sampling}
\label{sec:steaming_algorithm}

\subsection{The streaming algorithm}
\label{sec:alg_proof}

We present Algorithm \ref{alg:streamalgo} for unitary mixed Schur sampling. It takes as input $m+n$ qudits in state $\rho$ in a streaming manner and outputs a path $p_\gamma= (\gamma^1,...,\gamma^{m+n}=\gamma)$ and post measurement state $\rho^{m+n}\propto \Pi^{m,n}_{\gamma,p_{\gamma}}\rho \Pi^{m,n}_{\gamma,p_{\gamma}}$ for $\Pi^{m,n}_{\gamma,p_{\gamma}}$ as defined in Eq.~\eqref{eq:gpg_projections}.

\begin{figure}[h]
\begin{algorithm}[H]\label{alg:streamalgo}
\SetAlgoLined
\caption{A streaming algorithm for unitary mixed Schur sampling}
\SetKwInOut{Input}{Input}
\SetKwInOut{Output}{Output}
\SetKwInOut{Init}{Initialization}
\SetKwRepeat{Repeat}{repeat}{until}
\Input{A stream of $m+n$ qudits in state $\rho$}
\eIf{$m>0$}{
$\gamma^1=(1,0,...,0)$\,
}{
$\gamma^1=(0,...,0,-1)$\,
}
$\rho^1 = \rho$\,

\For{$k=1$ to $m+n-1$}{
Receive the $(k+1)$-th qudit to obtain the state $\tilde{\rho}^{k}$ on $\mathcal{Q}_{\lambda^{k}}^{d}\otimes \mathbb{C}^d$\;

\eIf{$k<m$}{
Apply $\tilde{\rho}^{k} \mapsto \sigma^{k+1}:= U_{\rm CG}^{\gamma^k}\rho^k(U_{\rm CG}^{\gamma^k})^\dagger$\,

Apply measurement $\sigma^{k+1}\mapsto \sum_{\nu\in\{\gamma^k+\ydiagram{1}\}}\ketbra{\nu}\otimes {\Pi}^{\gamma^k}_\nu\sigma^{k+1}{\Pi}^{\gamma^k}_\nu$ \,
}{
Apply $\tilde{\rho}^{k} \mapsto \sigma^{k+1}:= U_{\rm dCG}^{\gamma^k}\rho^k(U_{\rm dCG}^{\gamma^k})^\dagger$\,

Apply measurement $\sigma^{k+1}\mapsto \sum_{\nu\in\{\gamma^k-\ydiagram{1}\}}\ketbra{\nu}\otimes {\Pi}^{\gamma^k}_\nu\sigma^{k+1}{\Pi}^{\gamma^k}_\nu$ \,
}
Set $\gamma^{k+1} \leftarrow\nu$ upon observation of $\nu$ in measurement \,
Set $\tilde{\rho}^{k+1}:= \frac{{\Pi}^{\gamma^k}_{\gamma^{k+1}}\sigma^{k+1}{\Pi}^{\gamma^k}_{\gamma^{k+1}}}{\Tr[\sigma^{k+1}{\Pi}^{\gamma^k}_{\gamma^{k+1}}]}$\,
}
\Return $\rho^{m+n}, (\gamma^1,...,\gamma^{m+n})$
\end{algorithm}
\end{figure}

This streaming algorithm corresponds to a sequential application of the (dual) Clebsch-Gordan transform followed by a projective measurement onto the irreps of $SU(d)$, as given in  Eq.~\eqref{eq:cg_dcg}.
For each $k\in\{1,...,m+n\}$, the staircase $\gamma^{k}$ represents the irrep observed in the ensemble of the first $k$ qudits. We remark a few of things about Algorithm \ref{alg:streamalgo}:

\begin{itemize}
    \item Iteration $k\in\{1,..., m+n-1\}$ acts only on the $k+1$ qudits received thus far. However, for ease of exposition, we omit the $\mathbbm{I}_d^{\otimes(m+n-(k+1))}$ acting on the other qudits. This is the case for the unitaries, as well as the measurements we use.
    \item The measurements performed during the step $k$ correspond to the projections $\{\Pi^{\gamma^k}_{\nu}\}_{\nu\in\{\gamma^k+\ydiagram{1}\}}$ if $k<m$ and $\{\Pi^{\gamma^k}_{\nu}\}_{\nu\in\{\gamma^k-\ydiagram{1}\}}$ if $m\leq k<n$, from Eqs.~(\ref{eq:mid_projections},\ref{eq:mid_projections2}) respectively. 
    From the definitions of $U_{\rm CG}^{\gamma^k}$ and $U_{\rm dCG}^{\gamma^k}$ in Eq.~\eqref{eq:cg_dcg} it follows that they are complete projective measurements and thus well defined.
    \item Through a simple recursive argument using Born's rule and re-nomarlization, one can show that the output state of the algorithm is 
    \begin{align}\label{eq:pm_state}
        \rho^{m+n} = \frac{\left(\prod_{k=1}^{m+n-1}{\Pi}^{\gamma^k}_{\gamma^{k+1}}\cdot {U^{\gamma^k}}\right)\rho\left(\prod_{k=1}^{m+n-1}{\Pi}^{\gamma^k}_{\gamma^{k+1}}\cdot {U^{\gamma^k}}\right)^\dagger}{\Tr[\left(\prod_{k=1}^{m+n-1}{\Pi}^{\gamma^k}_{\gamma^{k+1}}\cdot {U^{\gamma^k}}\right)\rho\left(\prod_{k=1}^{m+n-1}{\Pi}^{\gamma^k}_{\gamma^{k+1}}\cdot {U^{\gamma^k}}\right)^\dagger]}\,.
    \end{align}
\end{itemize}

The following theorem tells us that Algorithm \ref{alg:streamalgo} correctly performs unitary mixed Schur sampling as defined in Section \ref{sec:results}.

\begin{lemma}[Correctness of Algorithm \ref{alg:streamalgo}]\label{lem:wmSsCorrectness}
    Let $\rho$ be an $m+n$ qudit state and let $p_{\Gamma}=(\Gamma^{1},...,\Gamma^{m+n}=\Gamma)$ and $\rho^{m+n}$ denote the output of Algorithm \ref{alg:streamalgo}.
    Then for all $\gamma \vdash_{d} (m,n)$:
    \begin{align}\label{eq:correct_prob}
        \Pr[\Gamma = \gamma] = \Tr[\big(U^{m,n}_{\rm Sch}\rho(U^{m,n}_{\rm Sch})^\dagger \big)\Pi^{m,n}_\gamma]\,,
    \end{align}
    where the probability is over the randomness in the Algorithm \ref{alg:streamalgo} and $\Pi^{m,n}_\gamma$ are the projections in Eq.~\eqref{eq:g_projections}.
    Further, it holds that
    \begin{align}\label{eq:correct_pm}
        \rho^{m+n} = \frac{\Tr_{\mathcal{P}_{\Gamma}^{d}}[(\Pi^{m,n}_{\Gamma, p_{\Gamma}}U_{\rm Sch}^{m,n})\rho(\Pi^{m,n}_{\Gamma, p_{\Gamma}}U_{\rm Sch}^{m,n})^\dagger]}{\Tr[\big(U_{\rm Sch}^{m,n}\rho(U_{\rm Sch}^{m,n})^\dagger\big)\Pi^{m,n}_{\Gamma, p_{\Gamma}}]}\,,
    \end{align}
    where $\Pi^{m,n}_{\Gamma, p_{\Gamma}}$ is defined in Eq.~\eqref{eq:gpg_projections}.
\end{lemma}

\begin{proof}    
    We fix a path of staircases $p_{\gamma}=(\gamma^1,...,\gamma=\gamma^{m+n})$ as the output $p_{\Gamma}$ of Algorithm \ref{alg:streamalgo}, and we use $p_{\gamma^k}$ to denote the partial paths $(\gamma^1,...,\gamma^k)$.

    First, we remark that for $k\in\{1,...,m-1\}$, Step $k$ of the algorithm corresponds to applying the operator $\Pi_{\gamma^{k+1}}^{\gamma^{k}} U^{\gamma^{k}}_{\rm CG}$, and then renormalizing. Similarly, for $k\in\{m,...,n-1\}$, Step $k$ corresponds to applying the operator $\Pi_{\gamma^{k+1}}^{\gamma^{k}} U^{\gamma^{k}}_{\rm dCG}$, and then renormalizing. 
    To prove the theorem we essentially want to show that the sequential application of these operators corresponds to
    \begin{align}
        \Pi_{\gamma,p_{\gamma}}^{m,n}U_{\rm Sch}^{m,n} \, .
    \end{align}
    We first take $k\in\{0,...,m-1\}$, and recall the definition of the super Clebsch-Gordan transform $U_{\rm CG}^{k,0}$ in Eq. \eqref{eq:super_CG} as
    \begin{align}
        U^{k,0}_{\rm CG}:=\sum_{\nu\vdash_{d} (k,0)}\mathbbm{I}_{\mathcal{P}_{\nu}^{d}}\otimes U^{\nu}_{\rm CG} \, .
    \end{align}
    If we restrict the input to the specific irrep labeled by $\gamma^{k}\vdash_{d}(k,0)$ and $p_{\gamma^{k}}$, then the action of $U_{\rm CG}^{k,0}$ is equivalent to the single Clebsch-Gordan transform $U_{\rm CG}^{\gamma^{k}}$. 
    Hence, Eq. \eqref{eq:super_CG} from Lemma \ref{lem:CG_Schur_proj_commute} is equivalent to
    \begin{align}
        \Pi_{\gamma^{k+1}}^{\gamma^{k}} U^{\gamma^{k}}_{\rm CG}(\Pi_{\gamma^{k},p_{\gamma^{k}}}^{k,0}U_{\rm Sch}^{k,0}\otimes\mathbbm{I}_{d})=\Pi_{\gamma^{k+1},p_{\gamma^{k+1}}}^{k+1,0}U_{\rm Sch}^{k+1,0} \, .
    \end{align}
    Here, we view the co-domains of the projections not as the full vector space, but as the restriction to their image\footnote{This effectively makes $\Pi_{\gamma^{k},p_{\gamma^{k}}}^{k,0}U_{\rm Sch}^{k,0}$ a co-isometry.}. 
    In a similar way, we obtain for $k\in\{m,...,n-1\}$ and $k=m+j$ that
    \begin{align}
        \Pi_{\gamma^{k+1}}^{\gamma^{k}} U^{\gamma^{k}}_{\rm dCG}(\Pi_{\gamma^{k},p_{\gamma^{k}}}^{m,j}U_{\rm Sch}^{m,j}\otimes\mathbbm{I}_{d})=\Pi_{\gamma^{k+1},p_{\gamma^{k+1}}}^{m,j+1}U_{\rm Sch}^{m,j+1} \,.
    \end{align}
    Applying this reasoning recursively gives us
    \begin{align}
        \Pi^{m,n}_{\gamma,\pg}U_{\rm Sch}^{m,n}
        &= \Big(\Pi_{\gamma}^{\gamma^{n+m-1}} U^{\gamma^{n+m-1}}_{\rm dCG}\Big)\cdot\Big(\Pi_{\gamma^{n+m-1},p_{\gamma^{n+m-1}}}^{m,n-1}U_{\rm Sch}^{m,n-1}\Big)\\
        &= \Big(\Pi_{\gamma}^{\gamma^{n+m-1}} U^{\gamma^{n+m-1}}_{\rm dCG}\Big)\cdot\Big(\Pi_{\gamma^{n+m-1}}^{\gamma^{n+m-2}} U^{\gamma^{n+m-2}}_{\rm dCG}\Big)\cdot\Big(\Pi_{\gamma^{n+m-2},p_{\gamma^{n+m-2}}}^{m,n-2}U_{\rm Sch}^{m,n-2}\Big)\\
        &~\,\vdots\\
        &= \prod_{k=n}^{m+n-1} \Big(\Pi^{\gamma^{k}}_{\gamma^{k+1}}U^{\gamma^k}_{\rm dCG}\Big)\cdot\Big(\Pi^{m,0}_{\gamma^m,p_{\gamma^m}}U^{m,0}_{\rm Sch}\Big)\\
        &~\,\vdots\\
        &=\prod_{k=n}^{m+n-1} \Big(\Pi^{\gamma^{k}}_{\gamma^{k+1}}U^{\gamma^k}_{\rm dCG}\Big)\cdot \prod_{k=1}^{m-1} \Big(\Pi^{\gamma^{k}}_{\gamma^{k+1}}U^{\gamma^k}_{\rm CG}\Big),
    \end{align}
    where we have again omitted identities in the $(m+n-(k+1))$ qudits which remain untouched during step $k$ in the recursion. Equality in Eq.~\eqref{eq:correct_pm} follows directly from the observation
    \begin{align}
        \Tr[\big(U_{\rm Sch}^{m,n}\rho(U_{\rm Sch}^{m,n})^\dagger\big)\Pi^{m,n}_{\gamma, \pg}] = \Tr[\big(\Pi^{m,n}_{\gamma, \pg}U_{\rm Sch}^{m,n}\big)\rho\big(\Pi^{m,n}_{\gamma, \pg}U_{\rm Sch}^{m,n})^\dagger\big)].
    \end{align}
    Equality of Eq.~\eqref{eq:correct_prob} follows by considering the choice of paths $p_\gamma$ traversed by Algorithm \ref{alg:streamalgo}.
    By (mixed) Schur-Weyl duality there are exactly $\dim \Pgd$ such paths, each observed with probability 
    \begin{align}
        \Tr[\left(\prod_{k=1}^{m+n-1}{\Pi}^{\gamma^k}_{\gamma^{k+1}}\cdot {U^{\gamma^k}}\right)\rho\left(\prod_{k=1}^{m+n-1}{\Pi}^{\gamma^k}_{\gamma^{k+1}}\cdot {U^{\gamma^k}}\right)^\dagger] = \Tr[\big(U_{\rm Sch}^{m,n}\rho(U_{\rm Sch}^{m,n})^\dagger\big)\Pi^{m,n}_{\gamma, \pg}].
    \end{align}
    Using the linearity of the trace we therefore obtain
    \begin{align}
        \Pr[\Gamma = \gamma]
        &= \sum_{\pg\in\Pgd}\Tr[\left(\prod_{k=1}^{m+n-1}{\Pi}^{\gamma^k}_{\gamma^{k+1}}\cdot {U^{\gamma^k}}\right)\rho\left(\prod_{k=1}^{m+n-1}{\Pi}^{\gamma^k}_{\gamma^{k+1}}\cdot {U^{\gamma^k}}\right)^\dagger]\\
        &= \sum_{\pg\in\Pgd}\Tr[\big(U_{\rm Sch}^{m,n}\rho(U_{\rm Sch}^{m,n})^\dagger\big)\Pi^{m,n}_{\gamma, \pg}]\\
        &= \Tr\Bigg[\big(U_{\rm Sch}^{m,n}\rho(U_{\rm Sch}^{m,n})^\dagger\big)\sum_{\pg\in\Pgd}\Pi^{m,n}_{\gamma, \pg}\Bigg]\\
        &= \Tr[\big(U_{\rm Sch}^{m,n}\rho(U_{\rm Sch}^{m,n})^\dagger\big)\Pi^{m,n}_{\gamma}],
    \end{align}
    where in the first inequality we let the labels ${\rm CG, dCG}$ of the unitaries be implied by the step $k$.
    This concludes the proof.
\end{proof}

\subsection{Gate and memory complexity}
\label{sec:complexity}
We define the gate complexity as the required number of elementary ($CNOT$, $H$, $Z^{\nicefrac{1}{4}}$) gates, and the memory complexity the maximum number of storage qubits needed at any point during the computation. For the (dual) Clebsch-Gordan transforms used in Algorithm \ref{alg:streamalgo}, we will use Algorithm \ref{alg:CG_algo} described in Appendix \ref{app:subs_algorithm_for_CG_transform}. This algorithm was first developed in \cite{HarrowTh05} and then generalized to the mixed Clebsch-Gordan transform in \cite{Nguyen_2023,Grinko_2023}.

First we remark that the encoding of states in Algorithm \ref{alg:CG_algo} allow a direct measurement of the irreps, so we do not incur additional gate or memory complexity from the measurement. The total gate complexity $T$ is now given by
\begin{align}
    T=\sum_{k=1}^{m}T_{\rm CG}^{k}+\sum_{k=1}^{n}T_{\rm dCG}^{k} \, ,
\end{align}
where $T_{\rm CG}^{k}$ and $T_{\rm dCG}^{k}$ is the gate complexity of applying the (dual) Clebsch-Gordan transform in the $k-$th step of Algorithm \ref{alg:streamalgo}. From Section \ref{app:CG_time_mem} we find that $T_{\rm CG}^{k}=T_{\rm dCG}^{k}=O(d^4\log_{2}^{p}(d,1/\epsilon'))$. Since gate errors are additive, we can replace $\epsilon'$ by $\epsilon/T$ to obtain total accuracy of $\epsilon$. Now summing over $k$ and inserting for $T$ we get
\begin{align}
    T=O((m+n)d^4\log_{2}^{p}(d,m,n,1/\epsilon)) \, .
\end{align}
For the memory complexity $M$, we similarly have
\begin{align}
    M=\max\{\max\limits_{1\leq k \leq m}M_{\rm CG}^{k},\max\limits_{1\leq k \leq n}M_{\rm dCG}^{k}\} \, ,
\end{align}
where $M_{\rm CG}^{k}$ and $M_{\rm dCG}^{k}$ are the memory requirements for storing the state and performing the Clebsch-Gordan transform in the $k-$th step of Algorithm \ref{alg:streamalgo}. From Section \ref{app:CG_time_mem} we get that $M_{\rm CG}^{k}=M_{\rm dCG}^{k}=O(d^2\log_{2}^{p}(d,1/\epsilon'))$, and again we get
\begin{align}
    O(d^2\log_{2}^{p}(d,m,n,1/\epsilon)) \, .
\end{align}
Together, this gives the following lemma.

\begin{lemma}
    \label{lem:memory_gate_complexity_1}
    The weak mixed Schur sampling given in Algorithm \ref{alg:streamalgo} can be achieved up to accuracy $\epsilon$ using a stream of $m+n$ input qudits, $M$ qubits of memory and $T$ elementary ($CNOT$, $H$, $Z^{\nicefrac{1}{4}}$) gates, where
    \begin{align}
        &M=O(d^2\log_{2}^{p}(d,m,n,1/\epsilon)) \, ,\\
        &T=O((m+n)d^4\log_{2}^{p}(d,m,n,1/\epsilon)) \, ,
    \end{align}
    with $p\approx 1.44$.
\end{lemma}

\subsection{Reduced dimension input}
In many applications of Schur-Weyl duality, such as quantum state tomography or spectrum estimation, special consideration is given to input states with reduced rank, e.g. a pure state with $r=1$. Therefore it is interesting to consider inputs on
\begin{align}
    K:=S^{\otimes m}\otimes\overline{S'}^{\otimes n} \, .
\end{align}
for subspaces $S,S'\subseteq\mathbbm{C}^{d}$ with $\dim S=r$ and $\dim S'=r'$.

We can deduce from Lemma \ref{lem:restriction_tensor_product_irreps} that all resulting staircases $\gamma$ have the property
\begin{align}
    \label{equ:memory_reduced_dimension_staircase_restriction}
    |\{i:\gamma_{i}>0\}|<r \quad , \quad |\{i:\gamma_{i}<0\}|<r' \, .
\end{align}
This means many entries in the encoding of states given in Appendix \ref{app:subs_algorithm_for_CG_transform} are just $\ket{0}$, and we can omit them. This immediately gives
\begin{align}
    M_{\rm CG}^{k}=M_{\rm dCG}^{k}=O((r+r')d\log_{2}^{p}(d,1/\epsilon')) \, .
\end{align}
Now we will investigate the gate complexity. We see that the (dual) Clebsch-Gordan transform links the staircases
\begin{align}
    (\gamma)' = \gamma\pm e_{j} \, ,
\end{align}
where $j\in[d]$. However, Eq. \ref{equ:memory_reduced_dimension_staircase_restriction} limits the possible $j$ to only $1\leq j \leq r$ or $d-r'+1\leq j \leq d$, since all other entries are $0$. Therefore, we can expect the relevant unitary matrices to actually be of size $\approx(r+r')\times(r+r')$ instead of $d\times d$. We can now apply this restriction to the argument given in \cite{Nguyen_2023} to obtain the gate complexity. This is described in more detail in Appendix \ref{app:reduced_dimension_input}, and we get
\begin{align}
    T_{\rm CG}^{k}=T_{\rm dCG}^{k}=O((r+r')^3d\log_{2}^{p}(d,1/\epsilon')) \, .
\end{align}
Setting again $\epsilon'=\epsilon/T$, taking the maximum for $M$ and summing over $k$ for $T$ then gives us

\begin{lemma}
    \label{lem:memory_gate_complexity_2}
    Let $S,S'\subseteq\mathbbm{C}^{d}$ be subspaces with $\dim S=r$ and $\dim S'=r'$. If the input of Algorithm \ref{alg:streamalgo} is restricted to the subspace
    \begin{align}
        S^{\otimes m}\otimes \overline{S'}^{\otimes n}\subseteq \Cdmn \, ,
    \end{align}
    then
    \begin{align}
        &M=O((r+r')d\log_{2}^{p}(d,m,n,1/\epsilon)) \, ,\\
        &T=O((m+n)(r+r')^3d\log_{2}^{p}(d,m,n,1/\epsilon)) \, ,
    \end{align}
    with $p\approx 1.44$.
\end{lemma}

\begin{remark}
    \label{rem:straightforward_implementation_memory}
    If the gate complexity is not important, we can instead use the estimate $\dim \mathcal{Q}_{\gamma}^{d}\leq (m+n)^{d(r+r')}$ and the naïve implementation of $U^{\gamma}_{\rm CG}$ and $U^{\gamma}_{\rm CG}$ through Solovay-Kitaev (see \cite{KiSt18, CM23}) to obtain 
    \begin{align}
        M=O((r+r')d\log_{2}(n+m)) \, .
    \end{align}
    However, the implementation of these unitaries then goes with time complexity $O(D^2\log_2(D/\epsilon))$, where $D$ is the Hilbert space dimension. This gives a gate complexity of
    \begin{align}
        T=O(d(r+r')(m+n)^{(2d(r+r')+1)}\log^{p}_2(m,n,1/\epsilon)) \, .
    \end{align}
\end{remark}

\section{Discussion}
Motivated by multiple existing applications of Schur-Weyl duality with permutation symmetry (e.g.  \cite{DoWr15,Haetal17,BlCrGo14,Buhrmanetal22}), we have introduced the task of \emph{unitary (mixed) Schur sampling} (Definitions \ref{def:unitary_Schur_sampling}, \ref{def:mixed_unitary_Schur_sampling}) and provided a streaming algorithm that performs this task (Algorithm \ref{alg:streamalgo}). We have analyzed the memory and time complexity of our algorithm (Theorem \ref{thm:memory_gate_complexity}), and we have shown an exponential advantage for the memory complexity in the constant $d$ regime when compared to implementing unitary (mixed) Schur sampling via existing algorithms (Tables \ref{tab:comparison_algorithms_unitary_Schur_sampling}, \ref{tab:comparison_similar_algorithms_unitary_Schur_sampling}, \ref{tab:comparison_algorithms_mixed_unitary_Schur_sampling}). In addition, we have obtained an improvement in time complexity by the factor $((r+r')/d)^{3}$ for the case where the input qudits are restricted to subspaces of dimension $r,r'\leq d$ (Theorem \ref{thm:memory_gate_complexity}). The same improvement also holds for the existing algorithms for the (mixed) Schur transform given in \cite{BCH05,Nguyen_2023,Grinko_2023} (Corollary \ref{cor:rank_improvement_Schur_transform}).

We expect our algorithm to be applicable in the context of any task with permutation (walled Brauer) symmetry that uses the (mixed) Schur transform. Examples of this are spectrum estimation and related problems \cite{KeyWer01,ChMi06,DoWr15}, state tomography \cite{DoWr16,Haetal17} (see also \cite{Cervero_2024}), as well as quantum majority vote \cite{Buhrmanetal22}, purification \cite{Cirac_99} and potential generalizations thereof. We further expect that the reduction of time complexity for input qudits on subspaces with dimension $r,r'\leq d$ will be useful for tasks such as spectrum estimation and state tomography, where reduced rank inputs are natural.

One central building block of our algorithm is the decomposition given in Eq. \eqref{equ:decomposition_Cxdxn}. This decomposition arises from repeated Clebsch-Gordan transforms and relies heavily on the representation theory of $SU(d)$. In \cite{Krovi19}, the author uses the representation theory of $\mathcal{S}_{m}$ to construct the Schur transform. It is possible that this algorithm can be approached in a similar manner to ours to obtain \textit{symmetric Schur sampling}, with potential improvements in gate and memory complexity.

Another potential avenue for future work is to investigate our algorithm in the regime $(m+n)\ll d$. We see from Table \ref{tab:comparison_algorithms_unitary_Schur_sampling} that in this regime, our algorithm potentially loses the advantage over the algorithm described in \cite{Krovi19}. In \cite{BaconThesis03}, an adjustment for this high-dimensional regime is discussed. An open question is whether such an adjustment could be made to our algorithm to preserve an advantage in memory and time complexity.

\section*{Acknowledgements}
We would like to thank Q.T. Nguyen for the helpful insights regarding the algorithm described in \cite{Nguyen_2023}.
ECM is supported by the National Research Foundation, Singapore and A*STAR under its CQT Bridging Grant. ET and LM are supported by ERC grant (QInteract, Grant No 101078107) and  VILLUM FONDEN (Grant No 10059 and 37532).

\bibliographystyle{IEEEtran}
\bibliography{references}

\newpage

\appendix
\section{Building the Clebsch-Gordan transform}
\label{app:building_the_CG_transform}
Similar to the Schur transform, we can split $U^{m,n}_{\rm CG}$ into smaller steps by reducing it to the Clebsch-Gordan transform for $SU(d-1)$. This was first described in \cite{HarrowTh05}, and \cite{Nguyen_2023,Grinko_2023} generalized it to mixed Clebsch-Gordan transforms.

\subsection{Reduced Wigner operators}
Here, we inspect the relation between the Clebsch-Gordan transforms of dimensions $d$ and $d-1$. In particular, let $\gamma$ be of length $d$ and consider the space $\mathcal{Q}_{\gamma}^{d}\otimes\mathbbm{C}^{d}$. First, we restate the partition of $\mathcal{Q}_{\gamma}^{d}$ into $SU(d)$ irreps, given by Eq. \eqref{equ:Gelfand-Tsetlin_decomposition}
\begin{align}
    \label{equ:Gelfand-Tsetlin_decomposition_appendix}
    \mathcal{Q}_{\gamma}^{d}\stackrel{SU(d-1)}{\cong}\bigoplus_{\nu\preceq\gamma}\mathcal{Q}_{\nu}^{d-1} \,.
\end{align}
Now we have two options to decompose $\mathcal{Q}_{\gamma}^{d}\otimes\mathbbm{C}^{d}$ into $SU(d-1)$ irreps. On the one hand, we apply $U_{\rm CG}^{\gamma}$, followed by the decomposition given in Eq. \eqref{equ:Gelfand-Tsetlin_decomposition_appendix}. We get
\begin{align}
    \label{equ:CG_then_GT_decomposition}
    \mathcal{Q}_{\gamma}^{d}\otimes\mathbbm{C}^{d}\stackrel{SU(d)}{\cong}\bigoplus_{\gamma'\in\{\gamma+\ydiagram{1}\}}\mathcal{Q}_{\gamma'}^{d}\stackrel{SU(d-1)}{\cong}\bigoplus_{\gamma'\in\{\gamma+\ydiagram{1}\}}\left(\bigoplus_{\nu'\preceq\gamma'}\mathcal{Q}_{\nu'}^{d-1}\right) \, .
\end{align}
On the other hand, since $\mathbbm{C}^{d}\cong\mathbbm{C}^{d-1}\oplus\mathbbm{C}$, swapping the order of operations by first using the decomposition in Eq. \eqref{equ:Gelfand-Tsetlin_decomposition_appendix} and then applying the Clebsch-Gordan transforms on $SU(d-1)$ gives
\begin{align}
    \label{equ:GT_then_CG_decomposition_1}
    \mathcal{Q}_{\gamma}^{d}\otimes\mathbbm{C}^{d}
    &\stackrel{SU(d-1)}{\cong}\left(\bigoplus_{\nu\preceq\gamma}\mathcal{Q}_{\nu}^{d-1}\otimes\mathbbm{C}^{d-1}\right)\oplus\left(\bigoplus_{\nu\preceq\gamma}\mathcal{Q}_{\nu}^{d-1}\right)\\
    \label{equ:GT_then_CG_decomposition_2}
    &\stackrel{SU(d-1)}{\cong}\left(\bigoplus_{\nu\preceq\gamma}\left(\bigoplus_{\nu'\in\{\nu+\ydiagram{1}\}}\mathcal{Q}_{\nu'}^{d-1}\right)\right)\oplus\left(\bigoplus_{\nu\preceq\gamma}\mathcal{Q}_{\nu}^{d-1}\right) \, .
\end{align}
Both cases start from $\mathcal{Q}_{\gamma}^{d}\otimes\mathbbm{C}^{d}$ and only involve operators that commute with the action of $SU(d-1)$, so we deduce that there exists an isometry between the right hand sides of Eqs.~\eqref{equ:CG_then_GT_decomposition} and \eqref{equ:GT_then_CG_decomposition_2} that commutes with the action of $SU(d-1)$. This is precisely the \emph{reduced Wigner operator} for the tensor product $\mathcal{Q}_{\gamma}^{d}\otimes\mathbbm{C}^{d}$, which we denote by $T_{\rm CG}^{d,\gamma}$.

We now want to apply Schur's lemma to the irreps of $SU(d-1)$. To this end, we label each irrep in Eqs. \eqref{equ:CG_then_GT_decomposition} and \eqref{equ:GT_then_CG_decomposition_2} according to the irrep they were derived from, to get
\begin{align}
    \label{equ:wigner_operator_mapping}
    T_{\rm CG}^{d,\gamma}:\left(\bigoplus_{\nu\preceq\gamma} \left(\bigoplus_{\nu'\in\{\nu+\ydiagram{1}\}} (\mathcal{Q}_{\nu'}^{d-1})_{\nu}\right)\right) \oplus \left(\bigoplus_{\nu\preceq\gamma}(\mathcal{Q}_{\nu}^{d-1})_{\nu}\right) \rightarrow \bigoplus_{\gamma'\in\{\gamma+\ydiagram{1}\}} \left(\bigoplus_{\nu'\preceq\gamma'} (\mathcal{Q}_{\nu'}^{d-1})_{\gamma'}\right) \, .
\end{align}
First we remark that the label distinguishes different equivalent irreps from one another, i.e. $(\mathcal{Q}_{\alpha}^{d-1})_{\beta}\cong\mathcal{Q}_{\alpha}^{d-1}$. Secondly, we know that $T_{\rm CG}^{d,\gamma}$ commutes with the action of $SU(d-1)$, so Schur's lemma implies that it is equivalent to scalar multiplication on each irrep. This means that if we fix $\gamma',\nu,\nu'$, we have the restriction $(T_{\rm CG}^{d,\gamma})_{\gamma',\nu}^{\nu'}$
\begin{align}
    \label{equ:T_gamma_mapping}
    (T_{\rm CG}^{d,\gamma})_{\gamma',\nu}^{\nu'}:(\mathcal{Q}_{\nu'}^{d-1})_{\nu}\rightarrow(\mathcal{Q}_{\nu'}^{d-1})_{\gamma'} \quad , \quad (T_{\rm CG}^{d,\gamma})_{\gamma',\nu}^{\nu'}=(t_{\rm CG}^{d,\gamma})_{\gamma',\nu}^{\nu'}\mathbbm{I}_{\mathcal{Q}_{\nu'}^{d-1}} \quad , \quad (t_{\rm CG}^{d,\gamma})_{\gamma',\nu}^{\nu'}\in\mathbbm{C} \, .
\end{align}
We now note that there exist $i\in[d]$ and $j\in[d-1]$ so that
\begin{align}
    \label{equ:relations_staircases_Wigner_Operator}
    \gamma'=\gamma+e_{i} \quad , \quad \nu'=\nu+e_{j} \quad \text{or} \quad \nu'=\nu \, .
\end{align}
The latter case corresponds to the second term in the direct sum in Eq. \eqref{equ:wigner_operator_mapping}, and we can denote it by $j=d$. This leads us to rewrite the operators in Eq. \eqref{equ:T_gamma_mapping} as
\begin{align}
    (T_{\rm CG}^{d,\gamma})_{\gamma',\nu}^{\nu'}\sim (T_{\rm CG}^{d,\gamma})_{i,j}^{\nu'} \quad , \quad (t_{\rm CG}^{d,\gamma})_{\gamma',\nu}^{\nu'}\sim (t_{\rm CG}^{d,\gamma})_{i,j}^{\nu'} \, .
\end{align}
For the dual Clebsch-Gordan transform, we can construct the operator $T_{\rm dCG}^{d,\gamma}$ in the same way, only now we have
\begin{align}
    \label{equ:relations_staircases_dual_Wigner_Operator}
    \gamma'=\gamma-e_{i} \quad , \quad \nu'=\nu-e_{j} \quad \text{or} \quad \nu'=\nu \, ,
\end{align}
and we denote the coefficients by $(t_{\rm dCG}^{d,\gamma})_{i,j}^{\nu'}\in\mathbbm{C}$. These scalars are given in an efficiently computable closed form expression in the next section.

\begin{remark}
    \label{rem:extend_matrix}
    It is important to note that a priori the above Eqs. are only defined for indices $i,j$ for which either Eqs. \eqref{equ:relations_staircases_Wigner_Operator} or \eqref{equ:relations_staircases_dual_Wigner_Operator} hold. However, we know that for all valid inputs $j$, they have to be isometries. For $(t_{\rm CG}^{d,\gamma})_{i,j}^{\nu'}$, this is due to the fact that the input spaces in Eqs. \eqref{equ:T_gamma_mapping} are orthogonal and $T_{\rm CG}^{d,\gamma}$ is an isometry. A similar argument also holds for $(t_{\rm dCG}^{d,\gamma})_{i,j}^{\nu'}$. Therefore we can extend both matrices to $d\times d$ unitaries. This fact will be important in Section \ref{sec:complexity}.    
\end{remark}

\subsection{Closed form expressions for the reduced Wigner operators}
\label{sec:closed_form_t}
The scalars $(t_{\rm CG}^{d,\gamma})_{i,j}^{\nu'}$ and $(t_{\rm dCG}^{d,\gamma})_{i,j}^{\nu'}$ are given in \cite[Volume 3, Chapter 18.2.10]{Vilenkin1995} by
\begin{align}
    \label{equ:Reduced_Wigner_d}
    &(t_{\rm CG}^{(d)})_{i,d}^{\gamma,\nu'}=\left|\frac{\prod\limits_{k=1}^{d-1}((\nu'_{k}-k)-(\gamma_{i}-i)-1)}{\prod\limits_{k\neq i}((\gamma_{k}-k)-(\gamma_{i}-i))}\right|^{\frac{1}{2}} \quad , \quad &, \\
    \label{equ:Reduced_Wigner_less_than_d}
    &(t_{\rm CG}^{d,\gamma})_{i,j}^{\nu'} = S(i,j) \left|\prod\limits_{k\neq j}\frac{(\nu'_{k}-k)-(\gamma_{i}-i)-1}{(\nu'_{k}-k)-(\nu'_{j}-j)}\prod\limits_{k\neq i}\frac{(\gamma_{k}-k)-(\nu'_{j}-j)+1}{(\gamma_{k}-k)-(\gamma_{i}-i)}\right|^{\frac{1}{2}} \quad , \quad \text{ if } j<d \, &,
\end{align}
\begin{align}
    \label{equ:Reduced_Wigner_dual_d}
    &(t_{\rm dCG}^{d,\gamma})_{i,j}^{\nu'}=\left|\frac{\prod\limits_{k=1}^{d-1}((\nu'_{k}-k)-(\gamma_{i}-i))}{\prod\limits_{k\neq i}((\gamma_{k}-k)-(\gamma_{i}-i))}\right|^{\frac{1}{2}} \quad , \quad &, \\
    \label{equ:Reduced_Wigner_dual_less_than_d}
    &(t_{\rm dCG}^{d,\gamma})_{i,j}^{\nu'} = S(i,j) \left|\prod\limits_{k\neq j}\frac{(\nu'_{k}-k)-(\gamma_{i}-i)}{(\nu'_{k}-k)-(\nu'_{j}-j)+1}\prod\limits_{k\neq i}\frac{(\gamma_{k}-k)-(\nu'_{j}-j)+2}{(\gamma_{k}-k)-(\gamma_{i}-i)}\right|^{\frac{1}{2}} \quad , \quad \text{ if } j<d \, &,
\end{align}
with
\begin{align}
    S(i,j)=1 \quad \text{ if } i\leq j \quad , \quad S(i,j)=-1 \quad \text{ if } i> j\, .
\end{align}

\subsection{A recursive construction of the Clebsch-Gordan transform}

The above considerations give us a recursive way to write the Clebsch-Gordan transform.
We begin by defining the operators $I_{\rm CG}^{d,\gamma}$ and $W_{\rm CG}^{d,\gamma}$
\begin{align}
    I_{\rm CG}^{d,\gamma}:\mathcal{Q}_{\gamma}^{d}\otimes\mathbbm{C}^{d}\rightarrow \left(\bigoplus_{\nu\preceq\gamma}\mathcal{Q}_{\nu}^{d-1}\otimes\mathbbm{C}^{d-1}\right)\oplus\left(\bigoplus_{\nu\preceq\gamma}\mathcal{Q}_{\nu}^{d-1}\right) \, ,
\end{align}
and
\begin{align}
\begin{split}
    W_{\rm CG}^{d,\gamma}&:\left(\bigoplus_{\nu\preceq\gamma}\mathcal{Q}_{\nu}^{d-1}\otimes\mathbbm{C}^{d-1}\right)\oplus\left(\bigoplus_{\nu\preceq\gamma}\mathcal{Q}_{\nu}^{d-1}\right) \rightarrow \left(\bigoplus_{\nu\preceq\gamma}\left(\bigoplus_{\nu'\in\{\nu+\ydiagram{1}\}}\mathcal{Q}_{\nu'}^{d-1}\right)\right)\oplus\left(\bigoplus_{\nu\preceq\gamma}\mathcal{Q}_{\nu}^{d-1}\right) \,, \\
    W_{\rm CG}^{d,\gamma}&:=\left(\bigoplus_{\nu\preceq\gamma}U_{\rm CG}^{\nu}\right)\oplus\left(\bigoplus_{\nu\preceq\gamma}\id_{\mathcal{Q}_{\nu}^{d-1}}\right) \, ,
\end{split}
\end{align}
which respectively realize the operations in Eqs.~\eqref{equ:GT_then_CG_decomposition_1} and \eqref{equ:GT_then_CG_decomposition_2}.
By construction, applying $I_{\rm CG}^{d,\gamma}$, followed by $W_{\rm CG}^{d,\gamma}$ and lastly $T_{\rm CG}^{d,\gamma}$ is equivalent to performing the Clebsch-Gordan transform of Eq.~\eqref{equ:CG_then_GT_decomposition}. 

We can define $I_{\rm dCG}^{d,\gamma},W_{\rm dCG}^{d,\gamma}$ in a similar way to $I_{\rm CG}^{d,\gamma},W_{\rm CG}^{d,\gamma}$, to obtain

\begin{lemma}
    \label{lem:iterative_CG}
    Let $\gamma\vdash_{d}(m,n)$. We have
    \begin{align}
        U_{\rm CG}^{\gamma}=T_{\rm CG}^{d,\gamma}W_{\rm CG}^{d,\gamma}I_{\rm CG}^{d,\gamma} \,, \\
        U_{\rm dCG}^{\gamma}=T_{\rm dCG}^{d,\gamma}W_{\rm dCG}^{d,\gamma}I_{\rm dCG}^{d,\gamma} \,.
    \end{align}
\end{lemma}

\begin{remark}
    While this may seem inconspicuous, we note that $W_{\rm CG}^{d,\gamma}$ and $W_{\rm dCG}^{d,\gamma}$ are defined via $U_{\rm CG}^{\nu}$ and $U_{\rm dCG}^{\nu}$. This means we can recur the above relation to obtain
    \begin{align}
        \label{equ:Recursion_U_CG}
        T_{\rm CG}^{d,\gamma}W_{\rm CG}^{d,\gamma}I_{\rm CG}^{d,\gamma}=T_{\rm CG}^{d,\gamma}\left(\bigoplus_{\nu\preceq\gamma}T_{\rm CG}^{d-1,\nu}W_{\rm CG}^{d-1,\nu}I_{\rm CG}^{d-1,\nu}\right)\oplus\left(\bigoplus_{\nu\preceq\gamma}\id_{\mathcal{Q}_{\nu}^{d-1}}\right)I_{\rm CG}^{d,\gamma} \,.
    \end{align}
    Since $W_{\rm CG}^{1,\gamma}$ is just the identity, we can therefore construct $U_{\rm CG}^{\gamma}$ by applying operators $T_{\rm CG}^{i,\gamma^{i}}$ and $I_{\rm CG}^{i,\gamma^{i}}$ for $i\in[d]$. We obtain a similar recursion for $T_{\rm dCG}^{d,\gamma}W_{\rm dCG}^{d,\gamma}I_{\rm dCG}^{d,\gamma}$. This insight was first presented in \cite{HarrowTh05}, and it is the second key part used to construct efficient algorithms for the (mixed) Schur transform in \cite{Bacon_2006,Nguyen_2023,Grinko_2023}.
\end{remark}

\section{An algorithm for the (dual) Clebsch-Gordan transform}
\label{app:algorithm_for_CG_transform}

\subsection{The algorithm}
\label{app:subs_algorithm_for_CG_transform}

We present the protocol for the Clebsch-Gordan transform based on the decomposition given in Lemma \ref{lem:iterative_CG}, and specifically Eq.~\eqref{equ:Recursion_U_CG}.
This algorithm was first developed in \cite{HarrowTh05} and then generalized to the mixed Clebsch-Gordan transform in \cite{Nguyen_2023,Grinko_2023}.

\begin{figure}[h]
    \begin{algorithm}[H]\label{alg:CG_algo}
        \SetAlgoLined
        \caption{An efficient algorithm for the Clebsch-Gordan transform}
        \SetKwInOut{Input}{Input}
        \SetKwInOut{Output}{Output}
        \SetKwInOut{Init}{Initialization}
        \SetKwRepeat{Repeat}{repeat}{until}
        \Input{The register $\ket{.}_{\mathcal{Q}_{\gamma}^{d}}$ encoding a state on $\mathcal{Q}_{\gamma}^{d}$ as in Eqs. \eqref{equ:encode_q_gamma} and \eqref{equ:encode_gamma_i}\\
        The register $\ket{.}_{\mathbbm{C}^{d}}$ encoding a state on $\mathbbm{C}^{d}$}

        Initialize auxiliary registers $\ket{0}_{A,{\rm aux}}$ and $\ket{0}_{B,{\rm aux}}$. \,
        
        \For{$1\leq i \leq d$}{

            \eIf{$i=1$}{
                Controlled by register $\ket{k}_{\mathbbm{C}^{d}}$:\,
                
                \eIf{$k\leq i$}{
                    Perform $\ket{\gamma^{1}}_{\mathcal{Q}_{\gamma}^{d}}\rightarrow \ket{\gamma^{1}+e_{1}}_{\mathcal{Q}_{\gamma}^{d}}(=\ket{(\gamma^{1})'}_{\mathcal{Q}_{\gamma}^{d}})$\,
                }{
                    Do nothing $(\text{i.e.} \ket{\gamma^{1}}_{\mathcal{Q}_{\gamma}^{d}}=\ket{(\gamma^{1})'}_{\mathcal{Q}_{\gamma}^{d}})$ \,
                }
            }{
                \textbf{Step 1:} Controlled by register $\ket{k}_{\mathbbm{C}^{d}}$:\,
                
                \eIf{$k\leq i$}{
                
                Controlled by register $\ket{q_{\gamma}}_{\mathcal{Q}_{\gamma}^{d}}$, use $(\gamma^{i-1})'$ and $\gamma^{i}$ to encode the matrix entries $(t_{\rm CG}^{i,\gamma^{i}})^{(\gamma^{i-1})'}_{j,j'}$ on register $\ket{0}_{A,{\rm aux}}$\,
                }{
                Encode the identity on register $\ket{0}_{A,{\rm aux}}$\,
                }

                \textbf{Step 2:} Controlled by register $\ket{t}_{A,{\rm aux}}$:\,

                Encode a sequence of gates on register $\ket{0}_{B,{\rm aux}}$ that approximate the matrix up to accuracy $\epsilon'$ by using the universal algorithm described in \cite[Chapter 4.5]{Nielsen_Chuang_2010} and the approximation of any $U\in SU(d)$ given in \cite{Kuperberg23}. \, 

                \textbf{Step 3:} Controlled by register $\ket{{\rm gates}}_{B,{\rm aux}}$:\,

                Apply the gate sequence to register $\ket{k}_{\mathbbm{C}^{d}}$\, 

                \textbf{Step 4:} Uncompute Step 2\,

                \textbf{Step 5:} Uncompute Step 1\,

                \textbf{Step 6:} Controlled by register $\ket{k}_{\mathbbm{C}^{d}}$:\,

                \eIf{$k\leq i$}{
                    Perform $\ket{\gamma^{i}}_{\mathcal{Q}_{\gamma}^{d}}\rightarrow \ket{\gamma^{i}+e_{k}}_{\mathcal{Q}_{\gamma}^{d}}(=\ket{(\gamma^{i})'}_{\mathcal{Q}_{\gamma}^{d}})$\,
                }{
                    Do nothing  $(\text{i.e.} \ket{\gamma^{i}}_{\mathcal{Q}_{\gamma}^{d}}=\ket{(\gamma^{i})'}_{\mathcal{Q}_{\gamma}^{d}})$ \,
                }
            }
        }
        \Return The registers $\ket{.}_{\mathcal{Q}_{\gamma}^{d}}$ and $\ket{.}_{\mathbbm{C}^{d}}$, encoding together a state on $\bigoplus_{k\in[d]}\mathcal{Q}_{\gamma+e_{k}}^{d}$
    \end{algorithm}
\end{figure}

The states $\ket{q_{\gamma}}\in\mathcal{Q}_{\gamma}^{d}$ are stored as a product 
\begin{align}
    \label{equ:encode_q_gamma}
    \ket{q_{\gamma}}=\ket{\gamma^{1}}\otimes...\otimes \ket{\gamma^{d}} \in (\mathbbm{C}^{n+m})^{\otimes d^2} \,,
\end{align}
with
\begin{align}
    \label{equ:encode_gamma_i}
    \ket{\gamma^{i}}=\ket{\gamma^{i}_{1}}\otimes ... \otimes \ket{\gamma^{i}_{i}}\otimes\ket{0} \otimes ... \otimes\ket{0} \in (\mathbbm{C}^{n+m})^{\otimes d} \,.
\end{align}
Figure \ref{fig:circuit_T} gives a graphical representation of one step of the algorithm. The isometries $I_{\rm CG}^{i,\gamma^{i}}$ from Eq.~\eqref{equ:Recursion_U_CG} are therefore trivial, since the register $\ket{q_{\gamma}}$ already has the desired decomposition. The operators $T_{\rm CG}^{i,\gamma^{i}}$ are applied via Steps 3 and 6, which correspond to
\begin{align}
    \ket{j}\rightarrow \sum_{j'}(t_{\rm CG}^{i,\gamma^{i}})^{(\gamma^{i-1})'}_{j,j'}\ket{j} \, ,
\end{align}
and
\begin{align}
    \ket{\gamma^{i}_{j'}}\rightarrow\ket{\gamma^{i}_{j'}+1} \, .
\end{align}
At the end of the algorithm, the register $\ket{j}$ tells us the new $SU(d)$ irrep
\begin{align}
    (\gamma^{d})'=\gamma^{d}+e_{j} \, .
\end{align}
We can obtain a similar algorithm for $U_{\rm dCG}^{\gamma}$ by calculating the matrix entries for $T_{\rm dCG}^{i,\gamma^{i}}$ instead.

\begin{remark}
    The uncomputations in Steps 4 and 5 are necessary to prevent entanglement with the auxiliary registers. Register $\ket{t}_{A,{\rm aux}}$ contains $d^2$ numerical values a $d\times d$ matrix, and register $\ket{\rm gates}_{B,{\rm aux}}$ contains a an encoding for the sequence of elementary gates needed to implement the matrix stored in $\ket{t}_{A,{\rm aux}}$. Both registers can be interpreted as classical encodings living on a quantum register. Finally, the accuracy $\epsilon'$ is given by $\epsilon/T^{\gamma}_{\rm (d)CG}$, where $T^{\gamma}_{\rm (d)CG}$ is the time complexity discussed in the next section.
\end{remark}

\begin{figure}
    \begin{center}
    \begin{quantikz}
        \lstick{$\ket{\gamma^{i}}$} & \ctrl{1} & & & & \ctrl{1} & \gate[1]{+} & \rstick{$\ket{(\gamma^{i})'}$} \\
        \lstick{$\ket{(\gamma^{i-1})'}$} & \ctrl{1} & & & & \ctrl{1} & & \rstick{$\ket{(\gamma^{i-1})'}$} \\
        \lstick{$\ket{j}$} & \ctrl{1} & & \gate[1]{t} & & \ctrl{1} & \ctrl{-2} & \rstick{$\ket{j'}$} \\
        \lstick{$\ket{0}_{A,{\rm aux}}$} & \gate[1]{C} & \ctrl{1} & & \ctrl{1} & \gate[1]{C^{\dagger}} & & \rstick{$\ket{0}_{A,{\rm aux}}$} \\
        \lstick{$\ket{0}_{B,{\rm aux}}$} & & \gate[1]{E} & \ctrl{-2} & \gate[1]{E^{\dagger}} & & & \rstick{$\ket{0}_{B,{\rm aux}}$}
    \end{quantikz}
    \end{center}
    \caption{Circuit diagram for Steps 1-6 of Algorithm \ref{alg:CG_algo}. $C$ denotes the classical calculation of the matrix entries $(t_{\rm CG}^{i,\gamma^{i}})^{(\gamma^{i-1})'}_{j,j'}$  given in Eqs. \eqref{equ:Reduced_Wigner_d} and \eqref{equ:Reduced_Wigner_less_than_d}. $E$ denotes the encoding as gates. $t$ corresponds to applying the encoded gates, and $+$ corresponds to the map $\gamma\rightarrow\gamma+e_{j}$. \label{fig:circuit_T}}
\end{figure}

\subsection{Gate and memory complexity}
\label{app:CG_time_mem}

Following the analysis in \cite{Nguyen_2023}, we now investigate the gate complexity $T_{\rm CG}^{\gamma}$ and the memory complexity $M_{\rm CG}^{\gamma}$ of implementing $U_{\rm CG}^{\gamma}$ with Algorithm \ref{alg:CG_algo} to precision $\epsilon$. The following individual steps then have the following gate complexities.
\begin{enumerate}
    \item The matrix entries $(t_{\rm CG}^{i,\gamma^{i}})^{(\gamma^{i-1})'}_{j,j'}$ take $O(d)$ classical computational steps to calculate. Altogether, there are at most $d^2$ entries. Extending the matrix to a unitary can be done via Gram-Schmidt, which takes $O(d^{3})$ classical computational steps. Each computational step can be implemented by using $O(1)$ gates, which gives a total complexity of $O(d^3)$ for this step.
    
    \item The decomposition of the matrix $(t_{\rm CG}^{i,\gamma^{i}})^{(\gamma^{i-1})'}_{j,j'}$ into elementary gates is a classical calculation that takes $O(d^2\log_{2}^{p}(1/\epsilon'))$ gates, where $p\approx1.44$ \cite{Kuperberg23}.
    
    \item Here we apply $O(d^2\log_{2}^{p}(1/\epsilon'))$ gates.
    
    \item Same complexity as Step 2.
    
    \item Same complexity as Step 1.
    
    \item The addition can be implemented in $O(d^2)$ gates.
    
\end{enumerate}
To apply $U_{\rm CG}^{\gamma}$ with accuracy $\epsilon$, we need to replace $\epsilon'$ by $\epsilon/T_{\rm CG}^{\gamma}$. This, and performing the above steps $d$ times, gives
\begin{align}\label{equ:CG_gate_complexity}
    T_{\rm CG}^{\gamma}=T_{\rm dCG}^{\gamma}=O(d^4\log_{2}^{p}(d,1/\epsilon)) \, .
\end{align}
The result for $T_{\rm dCG}^{\gamma}$ can be shown similarly as above.

For the memory complexity, we have to take into account the memory requirements $M_{\rm CG, \rho}^{\gamma}$ of storing the state and the size $M_{\rm CG, aux}^{\gamma}$ of the auxiliary registers. The encoding scheme in Eqs. \eqref{equ:encode_q_gamma} and \eqref{equ:encode_gamma_i} implies that register $\ket{.}_{\mathcal{Q}_{\gamma}^{d}}$ requires $\dim\Qgd\leq d^2\log_{2}(m+n)$ qubits, and register register $\ket{.}_{\mathbbm{C}^{d}}$ requires $\log_{2}(d)$ qubits. Therefore, to store the state itself the algorithm needs
\begin{align}
    M_{\rm CG, \rho}^{\gamma}=O(d^2\log_{2}(m+n)) \, ,
\end{align}
For the auxiliary memory, Step 1 tells us that register $\ket{.}_{A,{\rm aux}}$ requires $O(d^{2}\log_{2}(1/\epsilon))$ memory to store the matrix entries up to precision $O(\text{poly}(\epsilon))$. Step 2 further tells us that register $\ket{.}_{B,{\rm aux}}$ requires $O(d^2\log_{2}^{p}(d,n,m,1/\epsilon))$ memory to store all gates. Hence, together we have
\begin{align}
    M_{\rm CG, aux}^{\gamma}=O(d^2\log_{2}^{p}(1/\epsilon')) \, ,
\end{align}
auxiliary qubit registers. Summing $M_{\rm CG, \rho}^{\gamma}$ and $M_{\rm CG, aux}^{\gamma}$ and inserting for $\epsilon'=\epsilon/T_{\rm CG}^{\gamma}$, we get the memory complexity of Algorithm \ref{alg:CG_algo} as
\begin{align}
    M_{\rm CG}^{\gamma} = M_{\rm dCG}^{\gamma} = O(d^2\log_{2}^{p}(d,1/\epsilon))\,.
\end{align}
Again, we can show the result for $M_{\rm dCG}^{\gamma}$ similarly as above.

\subsection{Reduced dimension input}
\label{app:reduced_dimension_input}
In this section, we restrict the inputs to the set
\begin{align}
    K:=S^{\otimes m}\otimes\overline{S'}^{\otimes n} \, .
\end{align}
for subspaces $S,S'\subseteq\mathbbm{C}^{d}$ with $\dim S=r$ and $\dim S'=r'$.

For starters, note Lemma \ref{lem:mixed_SW_duality_restriction_partitions} from Appendix \ref{app:mixed_SW_duality_subspaces} implies that for inputs from $K$, all resulting staircases $\gamma$ in the mixed Schur-Weyl decomposition are such that
\begin{align}
    \label{equ:reduced_staircases}
    |\{i:\gamma_i> 0\}|\leq r \quad , \quad |\{i:\gamma_i< 0\}|\leq r' \, .
\end{align}
The same holds true for all intermediate steps in Algorithm \ref{alg:streamalgo} because the irreps correspond to the decomposition of $S^{\otimes k}$ for $1\leq k \leq m$ and $S^{\otimes m}\otimes\overline{S'}^{\otimes k}$ for $1\leq k \leq n$. Finally, the staircases that appear during Algorithm \ref{alg:CG_algo} also obey this rule, since they are $SU(i)$ irreps of the tensor powers of $S$, and so Lemma \ref{lem:mixed_SW_duality_restriction_partitions} holds as well. Altogether, this means that we can store each state $\ket{q_{\gamma}}\in\mathcal{Q}^{d}_{\gamma}$ as
\begin{align}
    \ket{q_{\gamma}}=\ket{\gamma^{1}}\otimes...\otimes \ket{\gamma^{d}} \in (\mathbbm{C}^{n+m})^{\otimes d(r+r'+1)} \,,
\end{align}
with
\begin{align}
    \ket{\gamma^{i}}=\ket{\gamma^{i}_{1}}\otimes...\otimes\ket{\gamma^{i}_{r}}\otimes \ket{0} \otimes \ket{\gamma^{i}_{d-r+1}} \otimes ... \otimes \ket{\gamma^{i}_{d}} \in (\mathbbm{C}^{n+m})^{\otimes (r+r'+1)} \,.
\end{align}
The $\ket{0}$ in the middle is a priori not necessary, but it will help with further calculations. We will now argue that $(t_{\rm CG}^{i,\gamma^{i}})^{(\gamma^{i-1})'}_{j,j'}$ is at most a $(r+r'+1)\times(r+r'+1)$ unitary, and the same holds for $(t_{\rm dCG}^{i,\gamma^{i}})_{j,j'}^{(\gamma^{i-1})'}$. 

We begin by investigating $T_{\rm CG}^{i,\gamma}$: Consider staircases
\begin{align}
    \nu&\vdash_{i-1}(m-1,n) \quad , \quad \nu'\vdash_{i-1}(m,n) \, , \\
    \gamma&\vdash_{i}(m-1,n) \quad , \quad \gamma'\vdash_{i}(m,n) \, , \\
\end{align}
with $j,j'\in [i]$ so that
\begin{align}
    \gamma'=\gamma+e_{j} \quad , \quad \nu'=\nu+e_{j'} \quad , \quad \nu\preceq\gamma \quad , \quad \nu'\preceq\gamma' \, .
\end{align}
We further assume that $\gamma,\nu,\nu'$ satisfy the conditions given in Eq. \eqref{equ:reduced_staircases}. Then $\gamma'$ can have at most $1$ nonzero entry more than $\gamma$, specifically if $j=r+1$ which forces $\gamma'_{r+1}=1$. Therefore we have
\begin{align}
    \label{eq:restriction_j_t_matrix}
    1\leq j \leq r+1 \quad \text{or} \quad d-r'+1\leq j \leq d \,.
\end{align}
In the same way we find that
\begin{align}
    \label{eq:restriction_jprime_t_matrix}
    1\leq j' \leq r+1 \quad \text{or} \quad d-r'+1\leq j' \leq d \,.
\end{align}
For $\nu'=\nu$, similar arguments hold. If we now restrict $\gamma,\nu'$ as discussed above, and for $j'$ in the range given in Eq.~\ref{eq:restriction_jprime_t_matrix}, we see that $(t_{\rm CG}^{i,\gamma^{i}})^{(\gamma^{i-1})'}_{j,j'}$ will be zero if $j$ is not given as in Eq.~\ref{eq:restriction_j_t_matrix}. Therefore we can view $(t_{\rm CG}^{i,\gamma^{i}})^{(\gamma^{i-1})'}_{j,j'}$ as a $(r+r'+1)\times(r+r'+1)$ matrix. If we return to Algorithm \ref{alg:CG_algo}, we can improve our analysis in the following way
\begin{enumerate}
    \item For calculating the matrix entries $(t_{\rm CG}^{i,\gamma^{i}})^{(\gamma^{i-1})'}_{j,j'}$, most factors in Eqs. \eqref{equ:Reduced_Wigner_d}, \eqref{equ:Reduced_Wigner_less_than_d}, \eqref{equ:Reduced_Wigner_dual_d} and \eqref{equ:Reduced_Wigner_dual_less_than_d} cancel out. Therefore it takes $O(r+r')$ classical computational steps to calculate a single entry. There are at most $(r+r'+1)^2$ matrix entries, and so Gram-Schmidt also takes $O((r+r')^{3})$ classical computational steps. Altogether we get a gate complexity of $O((r+r')^3)$ for this step.
    
    \item Since the matrix is now size $(r+r'+1)\times(r+r'+1)$, we need $O((r+r')^2\log_{2}^{p}(1/\epsilon'))$ gates.
    
    \item We apply the $O((r+r')^2\log_{2}^{p}(1/\epsilon'))$ gates.
    
    \item Same complexity as step 2.
    
    \item Same complexity as step 1.
    
    \item The addition can be implemented in $O((r+r')^2)$ gates.    
\end{enumerate}

Similarly to before, we now set $\epsilon'=\epsilon/T$ and repeat the above steps $d$ times to obtain
\begin{align}
    T_{\rm CG}^{\gamma}=T_{\rm dCG}^{\gamma}=O((r+r')^3d\log_{2}^{p}(d,1/\epsilon)) \, .
\end{align}
From Steps 2 and 3, we see that
\begin{align}
    M_{\rm CG, aux}^{\gamma}=O((r+r')^2\log_{2}^{p}(d,1/\epsilon)) \, .
\end{align}
For the state, the new encoding gives
\begin{align}
    M_{\rm CG, \rho}^{\gamma}=O(d(r+r')\log_{2}(m+n)) \, .
\end{align}
In total, we get the memory requirement and gate complexity of
\begin{align}
    &M_{\rm CG, \rho}^{\gamma}=M_{\rm dCG, \rho}^{\gamma}=O((r+r')d\log_{2}^{p}(d,1/\epsilon)) \, ,\\
    &T_{\rm CG}^{\gamma}=T_{\rm dCG}^{\gamma}=O((m+n)(r+r')^3d\log_{2}^{p}(d,1/\epsilon)) \, .
\end{align}
Here, the argument for $T_{\rm dCG}^{\gamma}$ and $M_{\rm dCG}^{\gamma}$ is the same as above.

\section{Mixed Schur-Weyl duality for subspaces}
\label{app:mixed_SW_duality_subspaces}

Instead of looking at the full mixed Schur-Weyl duality over $(\mathbb{C}^d)^{\otimes m}\otimes(\overline{\mathbbm{C}}^d)^{\otimes n}$, we might just be interested in a subspace
\begin{align}
    K\subseteq \Cdmn\stackrel{\mathcal{A}_{m,n}^{d}\times SU(d)}{\cong} 
    \bigoplus_{\gamma\vdash_{d} (m,n)}\mathcal{P}_{\gamma}^{d}\otimes\mathcal{Q}_{\gamma}^d \, .
\end{align}
It is possible that not all $\gamma\vdash_{d} (m,n)$ are relevant for the above decomposition of $K$ so in this section we explore the decomposition resulting from $U_{\rm mSch}^{m,n}(K)$. First, we will lay the groundwork by investigating which irreps can appear in the tensor product $\mathcal{Q}_{\lambda}^{d}\otimes\overline{\mathcal{Q}_{\mu}^{d}}$, and afterwards we will discuss mixed Schur-Weyl duality for subspaces.

\begin{remark}
    In the following, we will frequently use the number of nonzero entries of a given partition $\lambda\vdash_{d}m$. We denote this number as the \textit{length} $l(\lambda)$. It is important to note that $1\leq l(\lambda)\leq d$.
\end{remark}

\subsection{Irreps of $SU(d)$ in $\mathcal{Q}_{\lambda}^{d}\otimes\overline{\mathcal{Q}_{\mu}^{d}}$}
\label{subs:irreps_of_SUd_in_tensor_product}

We start with the following definition.

\begin{definition}
    Let $m,n,d\in\mathbbm{N}$ and let $\lambda\vdash_{d}m$ and $\mu\vdash_{d}n$. Then we define the \textit{Littlewood-Richardson Coefficient} $c^{\nu}_{\lambda,\mu}\in\mathbbm{N}_{0}$ as the multiplicity of the irrep $\mathcal{Q}_{\nu}^{d}$ inside the decomposition
    \begin{align}
        \mathcal{Q}_{\lambda}^{d}\otimes\mathcal{Q}_{\mu}^{d} \stackrel{SU(d)}{\cong} \bigoplus_{\nu}\mathbbm{C}^{c^{\nu}_{\lambda,\mu}}\otimes\mathcal{Q}_{\nu}^{d} \, .
    \end{align}    
\end{definition}

These coefficients can be determined by the \textit{Littlewood-Richardsen Rule}, given f.ex. in \cite[Chapter 4.9]{Sagan_2011}. The following lemma summarizes some consequences.

\begin{lemma}
    \label{lem:littlewood-richardson_rules}
    Let $\lambda\vdash_{d}m$ and $\mu\vdash_{d}n$. Then we have
    \begin{enumerate}
        \setcounter{enumi}{0}
        \item $c^{\nu}_{\lambda,\mu}=c^{\nu}_{\mu,\lambda}$.
    \end{enumerate}
    Let further $c^{\nu}_{\lambda,\mu}\geq1$. Then we have
    \begin{enumerate}
        \setcounter{enumi}{1}
        \item $\nu\vdash_{d}(m+n)$,
        \item $\nu_{i}\geq\lambda_{i}$ for all $i\in[d]$.
    \end{enumerate}
    Let further $j,k\in\mathbbm{N}$ be such that $\nu_{i+k-1}>\lambda_{i}$. Then we have
    \begin{enumerate}
        \setcounter{enumi}{3}
        \item $l(\mu)\geq k$.
    \end{enumerate}
\end{lemma}

\begin{proof}
    The proof requires deeper familiarity about partitions and their relation to so-called (Young) tableaux. For the sake of brevity, we assume familiarity with the Littlewood-Richardsen rule. Statement 1 follows from commutativity of the tensor product, and Statements 2 and 3 follow directly from the Littlewood-Richardsen rule. Statement 4 follows from the requirement that there has to exist a semistandard tableau with shape $\nu/\lambda$ and with content $\mu$. A semistandard tableau with content $\mu$ can have colums of at most size $l(\mu)$. However, $\nu_{i+k}>\lambda_{i}$ implies that $\nu/\lambda$ has a column of at least size $k$.
\end{proof}

In addition, we need some facts from the representation theory of $SU(d)$ and the connection to partitions and staircases. More information about this subject can be found in \cite[Chapter 3, Chapter 9.1]{Goodman2009}.

\begin{fact}
    \label{lem:staircase_add_column}
    Let $\gamma\vdash_{d}(m,n)$ be a staircase, let $k\in\mathbbm{Z}$ and let
    \begin{align}
        \nu:=(\gamma_{1}+k,\gamma_{2}+k,...,\gamma_{d}+k) \, .
    \end{align}
    Then $\nu$ is again a staircase and we have
    \begin{align}
        \mathcal{Q}_{\gamma}^{d} \stackrel{SU(d)}{\cong}\mathcal{Q}_{\nu}^{d} \, .
    \end{align}
    Further, if $\gamma\vdash_{d}(m,0)$ and $\lambda\vdash_{d}m$ is given by $\lambda_{i}=\gamma_{i}$ for all $i\in[d]$, then we have
    \begin{align}
        \mathcal{Q}_{\gamma}^{d} \stackrel{SU(d)}{\cong}\mathcal{Q}_{\lambda}^{d} \, .
    \end{align}
\end{fact}

\begin{fact}
    \label{lem:partition_dual_irrep}
    Let $\mu\vdash_{d}m$ be a partition, and let
    \begin{align}
        \overline{\mu}:=(\mu_{1}-\mu_{d},\mu_{1}-\mu_{d-1},...,\mu_{1}-\mu_{2},0) \, .
    \end{align}
    Then we have
    \begin{align}
        \overline{\mathcal{Q}_{\mu}^{d}} \stackrel{SU(d)}{\cong}\mathcal{Q}_{\overline{\mu}}^{d} \, .
    \end{align}
\end{fact}

Now we have the tools to prove the following lemma

\begin{lemma}
    \label{lem:restriction_tensor_product_irreps}
    Let $m,n,d\in\mathbbm{N}$ and let $\lambda\vdash_{d}m$ and $\mu\vdash_{d}n$. Let further $l(\lambda)=r$ and $l(\mu)=r'$. Then we have
    \begin{align}
        \label{equ:lem_restriction_tensor_product_irreps_1}        \mathcal{Q}_{\lambda}^{d}\otimes\overline{\mathcal{Q}_{\mu}^{d}} \stackrel{SU(d)}{\cong} \bigoplus_{\gamma\vdash_{d}(m,n)}\mathbbm{C}^{c^{\gamma}_{\lambda,\overline{\mu}}}\otimes\mathcal{Q}_{\gamma}^{d} \, ,
\end{align}
    with $c^{\nu}_{\lambda,\mu}\in\mathbbm{N}_{0}$ and
    \begin{align}
        \label{equ:lem_restriction_tensor_product_irreps_2}
        c^{\gamma}_{\lambda,\overline{\mu}} >0 \quad \Rightarrow \quad |\{i:\gamma_{i}>0\}|\leq r \quad \text{and} \quad  |\{i:\gamma_{i}<0\}|\leq r' \, .
    \end{align}
\end{lemma}

\begin{proof}
    Lemma \ref{lem:partition_dual_irrep} tells us that
    \begin{align}
        \label{equ:proof_tensor_products_into_irreps}
        \mathcal{Q}_{\lambda}^{d}\otimes\overline{\mathcal{Q}_{\mu}^{d}} \stackrel{SU(d)}{\cong}\mathcal{Q}_{\lambda}^{d}\otimes\mathcal{Q}_{\overline{\mu}}^{d} \stackrel{SU(d)}{\cong}\bigoplus_{\nu}\mathbbm{C}^{c^{\nu}_{\lambda,\overline{\mu}}}\otimes\mathcal{Q}_{\nu}^{d} \, ,
    \end{align}
     where the last equivalence is the decomposition into irreps with multiplicities $c^{\nu}_{\lambda,\overline{\mu}}$. Since $\mu$ has only $r'$ nonzero entries, we find that
    \begin{align}
        \overline{\mu}=(\mu_{1},...,\mu_{1},\mu_{1}-\mu_{r'},\mu_{1}-\mu_{r'-1},...,0) \, .
    \end{align}
    Now Lemma \ref{lem:littlewood-richardson_rules} tells us that $c^{\nu}_{\lambda,\overline{\mu}}=c^{\nu}_{\overline{\mu},\lambda}$. Further, for $c^{\nu}_{\overline{\mu},\lambda}\geq1$ we have 
    \begin{enumerate}
        \item $\nu\vdash_{d}(d\cdot\mu_{1}-n)+m$,
        \item $\nu_{d-r'}\geq\overline{\mu}_{d-r'}=\mu_{1}$,
        \item $\nu_{r+1}\leq \overline{\mu}_{1}=\mu_{1}$.
    \end{enumerate}
    The first point comes from the fact that $\overline{\mu}\vdash_{d}(d\cdot\mu_{1}-n)$, and the second point follows from the fact that $\nu_{i}\geq\overline{\mu}_{i}$ for all $i\in[d]$. For the third point, we remember that $\nu_{r+1}>\overline{\mu}_{1}$ implies that $l(\lambda)\geq r+1$. However, our assumption was that $l(\lambda)=r$. It is important to note that the above statements hold for all partitions $\nu$, whose irreps have nontrivial contribution to Eq. \eqref{equ:proof_tensor_products_into_irreps}. We now use Lemma \ref{lem:staircase_add_column} to obtain equivalent staircases $\gamma\vdash_{d}(m,n)$. For a given $\nu$, we take
    \begin{align}
        \gamma=(\nu_{1}-\mu_{1},...,\nu_{d}-\mu_{1}) \, .
    \end{align}
    It is straightforward to verify that $\gamma\vdash_{d}(m,n)$, and that all $\gamma\vdash_{d}(m,n)$ with nontrivial contribution in Eq. \eqref{equ:lem_restriction_tensor_product_irreps_1} are of this form. The inequalities in Eq. \eqref{equ:lem_restriction_tensor_product_irreps_2} now follow directly from the facts that $\nu_{d-r'}\geq\mu_{1}$, that $\nu_{r+1}\leq\mu_{1}$, and that the entries $\nu_{i}$ are monotone decreasing.
\end{proof}

\subsection{Restriction on staircases}

We start by inspecting Schur-Weyl duality, that is $n=0$. Let $S\subseteq\mathbbm{C}^{d}$ be a subspace with $\dim S=r$, and let
\begin{align}
    K:=S^{\otimes m} \, .
\end{align}
We can now take the well-known result (see e.g. \cite[Lemma 2]{Haetal17})

\begin{lemma}
    Let $\rho$ be a state on $\mathbbm{C}^{d}$ with rank $r$, and let $\lambda\vdash_{d}m$ with $l(\lambda)>r$. Then we have
    \begin{align}
        \Tr[\rho^{\otimes m}\Pi_{\lambda}^{m}]=0 \, .
    \end{align}
\end{lemma}

We remark that if $\rho$ is the maximally mixed state on $S$, then we have $\id_{K}=r^m\rho^{\otimes m}$, which tells us that
\begin{align}
    \label{equ:SW_duality_restriction_partitions}
    U_{\rm Sch}^{m}(K)\subseteq \bigoplus_{\lambda\vdash_{d}m \,, \, l(\lambda)\leq r}\mathcal{P}_{\lambda}\otimes\mathcal{Q}_{\lambda}^{d} \, .
\end{align}

For mixed Schur-Weyl duality, we can obtain a similar result through the following considerations. First, let $S,S'\subseteq\mathbbm{C}^{d}$ with $\dim S=r$ and $\dim S'=r'$, and set
\begin{align}
    K:=S^{\otimes m}\otimes\overline{S'}^{\otimes n} \, .
\end{align}
We can apply Eq. \ref{equ:SW_duality_restriction_partitions} to both the normal and the dual tensor power to obtain
\begin{align}
    \label{equ:subspace_S_m_S_n_irreps}
    (U^{m}_{\rm Sch}\otimes \overline{U^{n}_{\rm Sch}})(K)\subseteq \bigoplus_{\substack{\lambda\vdash_{d}m \,, \, l(\lambda)\leq r \\ \mu\vdash_{d}n \,, \, l(\mu)\leq r'}} \mathcal{P}_{\lambda}\otimes\overline{\mathcal{P}_{\mu}}\otimes\mathcal{Q}_{\lambda}^{d}\otimes\overline{\mathcal{Q}_{\mu}^{d}}\, .
\end{align}
Next, we want to find out which $SU(d)$ irreps appear in Eq.~\eqref{equ:subspace_S_m_S_n_irreps}. To this end, we use Lemma \ref{lem:restriction_tensor_product_irreps} from Appendix \ref{subs:irreps_of_SUd_in_tensor_product}, which tells us that
\begin{align}
    \mathcal{Q}_{\lambda}^{d}\otimes\overline{\mathcal{Q}_{\mu}^{d}} \stackrel{SU(d)}{\cong} \bigoplus_{\gamma\vdash_{d}(m,n)}\mathbbm{C}^{c^{\gamma}_{\lambda,\overline{\mu}}}\otimes\mathcal{Q}_{\gamma}^{d} \, ,
\end{align}
with $c^{\gamma}_{\lambda,\overline{\mu}}\in\mathbbm{N}_{0}$ and
\begin{align}
        c^{\gamma}_{\lambda,\overline{\mu}} >0 \quad \Rightarrow \quad |\{i:\gamma_{i}>0\}|\leq r \quad \text{and} \quad  |\{i:\gamma_{i}<0\}|\leq r' \, .
    \end{align}
If we now apply the mixed Schur transform instead of two Schur transforms on the left hand side of Eq.~\eqref{equ:subspace_S_m_S_n_irreps}, we obtain
\begin{align}
    \label{equ:subspace_A_m,n_irreps}
    U_{\rm Sch}^{m,n}(K)\subseteq \bigoplus_{\gamma\vdash_{d} (m,n)}\mathcal{P}_{\gamma}^{d}\otimes\mathcal{Q}_{\gamma}^d \, .
\end{align}
By Schur's Lemma, the $SU(d)$ irreps in Eqs. \eqref{equ:subspace_S_m_S_n_irreps} and \eqref{equ:subspace_A_m,n_irreps} have to be the same. Therefore we can use Lemma \ref{lem:restriction_tensor_product_irreps} to omit all those with more than $r$ positive or $r'$ negative entries, and we arrive at the following result

\begin{lemma}\label{lem:mixed_SW_duality_restriction_partitions}
    We have
    \begin{align}
        U_{\rm Sch}^{m,n}(K)\subseteq \bigoplus_{\substack{\gamma\vdash_{d}(m,n): \\ |\{i:\gamma_{i} > 0\}|\leq r \\ |\{i:\gamma_{i} < 0\}|\leq r'}}\mathcal{P}_{\gamma}^{d}\otimes\mathcal{Q}_{\gamma}^{d} \, .
    \end{align}
\end{lemma}

\end{document}

%% file: wBrauer_tikz.tex
\tikzset{every picture/.style={line width=0.75pt}} %set default line width to 0.75pt        

\begin{tikzpicture}[x=0.75pt,y=0.75pt,yscale=-1,xscale=1]
%uncomment if require: \path (0,454); %set diagram left start at 0, and has height of 454

%Curve Lines [id:da9199092626836516] 
\draw    (50,200) .. controls (86,208) and (127,242) .. (130,250) ;
\draw [shift={(130,250)}, rotate = 69.44] [color={rgb, 255:red, 0; green, 0; blue, 0 }  ][fill={rgb, 255:red, 0; green, 0; blue, 0 }  ][line width=0.75]      (0, 0) circle [x radius= 3.35, y radius= 3.35]   ;
\draw [shift={(50,200)}, rotate = 12.53] [color={rgb, 255:red, 0; green, 0; blue, 0 }  ][fill={rgb, 255:red, 0; green, 0; blue, 0 }  ][line width=0.75]      (0, 0) circle [x radius= 3.35, y radius= 3.35]   ;
%Curve Lines [id:da7267122251095792] 
\draw    (130,200) .. controls (108.5,216.5) and (101.5,227.5) .. (90,250) ;
\draw [shift={(90,250)}, rotate = 117.07] [color={rgb, 255:red, 0; green, 0; blue, 0 }  ][fill={rgb, 255:red, 0; green, 0; blue, 0 }  ][line width=0.75]      (0, 0) circle [x radius= 3.35, y radius= 3.35]   ;
\draw [shift={(130,200)}, rotate = 142.5] [color={rgb, 255:red, 0; green, 0; blue, 0 }  ][fill={rgb, 255:red, 0; green, 0; blue, 0 }  ][line width=0.75]      (0, 0) circle [x radius= 3.35, y radius= 3.35]   ;
%Curve Lines [id:da5217886873212119] 
\draw    (170,200) .. controls (190.5,218) and (197,223) .. (210,250) ;
\draw [shift={(210,250)}, rotate = 64.29] [color={rgb, 255:red, 0; green, 0; blue, 0 }  ][fill={rgb, 255:red, 0; green, 0; blue, 0 }  ][line width=0.75]      (0, 0) circle [x radius= 3.35, y radius= 3.35]   ;
\draw [shift={(170,200)}, rotate = 41.28] [color={rgb, 255:red, 0; green, 0; blue, 0 }  ][fill={rgb, 255:red, 0; green, 0; blue, 0 }  ][line width=0.75]      (0, 0) circle [x radius= 3.35, y radius= 3.35]   ;
%Curve Lines [id:da06327809699802045] 
\draw    (210,200) .. controls (194,216.5) and (184.5,224.5) .. (170,250) ;
\draw [shift={(170,250)}, rotate = 119.62] [color={rgb, 255:red, 0; green, 0; blue, 0 }  ][fill={rgb, 255:red, 0; green, 0; blue, 0 }  ][line width=0.75]      (0, 0) circle [x radius= 3.35, y radius= 3.35]   ;
\draw [shift={(210,200)}, rotate = 134.12] [color={rgb, 255:red, 0; green, 0; blue, 0 }  ][fill={rgb, 255:red, 0; green, 0; blue, 0 }  ][line width=0.75]      (0, 0) circle [x radius= 3.35, y radius= 3.35]   ;
%Curve Lines [id:da3849052898383889] 
\draw    (90,200) .. controls (68.5,216.5) and (61.5,227.5) .. (50,250) ;
\draw [shift={(50,250)}, rotate = 117.07] [color={rgb, 255:red, 0; green, 0; blue, 0 }  ][fill={rgb, 255:red, 0; green, 0; blue, 0 }  ][line width=0.75]      (0, 0) circle [x radius= 3.35, y radius= 3.35]   ;
\draw [shift={(90,200)}, rotate = 142.5] [color={rgb, 255:red, 0; green, 0; blue, 0 }  ][fill={rgb, 255:red, 0; green, 0; blue, 0 }  ][line width=0.75]      (0, 0) circle [x radius= 3.35, y radius= 3.35]   ;
%Straight Lines [id:da3003468129179352] 
\draw  [dash pattern={on 4.5pt off 4.5pt}]  (110,190) -- (110,260) ;
%Curve Lines [id:da7793673539290189] 
\draw    (290,200) .. controls (310.5,214) and (359,214.5) .. (370,200) ;
\draw [shift={(370,200)}, rotate = 307.18] [color={rgb, 255:red, 0; green, 0; blue, 0 }  ][fill={rgb, 255:red, 0; green, 0; blue, 0 }  ][line width=0.75]      (0, 0) circle [x radius= 3.35, y radius= 3.35]   ;
\draw [shift={(290,200)}, rotate = 34.33] [color={rgb, 255:red, 0; green, 0; blue, 0 }  ][fill={rgb, 255:red, 0; green, 0; blue, 0 }  ][line width=0.75]      (0, 0) circle [x radius= 3.35, y radius= 3.35]   ;
%Curve Lines [id:da7782871821788435] 
\draw    (370,250) .. controls (363,240.5) and (337,240.5) .. (330,250) ;
\draw [shift={(330,250)}, rotate = 126.38] [color={rgb, 255:red, 0; green, 0; blue, 0 }  ][fill={rgb, 255:red, 0; green, 0; blue, 0 }  ][line width=0.75]      (0, 0) circle [x radius= 3.35, y radius= 3.35]   ;
\draw [shift={(370,250)}, rotate = 233.62] [color={rgb, 255:red, 0; green, 0; blue, 0 }  ][fill={rgb, 255:red, 0; green, 0; blue, 0 }  ][line width=0.75]      (0, 0) circle [x radius= 3.35, y radius= 3.35]   ;
%Curve Lines [id:da6400616929793119] 
\draw    (410,200) .. controls (430.5,218) and (437,223) .. (450,250) ;
\draw [shift={(450,250)}, rotate = 64.29] [color={rgb, 255:red, 0; green, 0; blue, 0 }  ][fill={rgb, 255:red, 0; green, 0; blue, 0 }  ][line width=0.75]      (0, 0) circle [x radius= 3.35, y radius= 3.35]   ;
\draw [shift={(410,200)}, rotate = 41.28] [color={rgb, 255:red, 0; green, 0; blue, 0 }  ][fill={rgb, 255:red, 0; green, 0; blue, 0 }  ][line width=0.75]      (0, 0) circle [x radius= 3.35, y radius= 3.35]   ;
%Curve Lines [id:da9573072874716506] 
\draw    (450,200) .. controls (434,216.5) and (424.5,224.5) .. (410,250) ;
\draw [shift={(410,250)}, rotate = 119.62] [color={rgb, 255:red, 0; green, 0; blue, 0 }  ][fill={rgb, 255:red, 0; green, 0; blue, 0 }  ][line width=0.75]      (0, 0) circle [x radius= 3.35, y radius= 3.35]   ;
\draw [shift={(450,200)}, rotate = 134.12] [color={rgb, 255:red, 0; green, 0; blue, 0 }  ][fill={rgb, 255:red, 0; green, 0; blue, 0 }  ][line width=0.75]      (0, 0) circle [x radius= 3.35, y radius= 3.35]   ;
%Curve Lines [id:da593660325370754] 
\draw    (330,200) .. controls (308.5,216.5) and (301.5,227.5) .. (290,250) ;
\draw [shift={(290,250)}, rotate = 117.07] [color={rgb, 255:red, 0; green, 0; blue, 0 }  ][fill={rgb, 255:red, 0; green, 0; blue, 0 }  ][line width=0.75]      (0, 0) circle [x radius= 3.35, y radius= 3.35]   ;
\draw [shift={(330,200)}, rotate = 142.5] [color={rgb, 255:red, 0; green, 0; blue, 0 }  ][fill={rgb, 255:red, 0; green, 0; blue, 0 }  ][line width=0.75]      (0, 0) circle [x radius= 3.35, y radius= 3.35]   ;
%Straight Lines [id:da2385502138937836] 
\draw  [dash pattern={on 4.5pt off 4.5pt}]  (350,190) -- (350,260) ;
%Straight Lines [id:da9690220791572071] 
\draw    (230,223.5) -- (261,223.5)(230,226.5) -- (261,226.5) ;
\draw [shift={(270,225)}, rotate = 180] [fill={rgb, 255:red, 0; green, 0; blue, 0 }  ][line width=0.08]  [draw opacity=0] (8.93,-4.29) -- (0,0) -- (8.93,4.29) -- cycle    ;

\end{tikzpicture}

%% file: main.bbl
\begin{thebibliography}{10}
\providecommand{\url}[1]{#1}
\csname url@samestyle\endcsname
\providecommand{\newblock}{\relax}
\providecommand{\bibinfo}[2]{#2}
\providecommand{\BIBentrySTDinterwordspacing}{\spaceskip=0pt\relax}
\providecommand{\BIBentryALTinterwordstretchfactor}{4}
\providecommand{\BIBentryALTinterwordspacing}{\spaceskip=\fontdimen2\font plus
\BIBentryALTinterwordstretchfactor\fontdimen3\font minus \fontdimen4\font\relax}
\providecommand{\BIBforeignlanguage}[2]{{%
\expandafter\ifx\csname l@#1\endcsname\relax
\typeout{** WARNING: IEEEtran.bst: No hyphenation pattern has been}%
\typeout{** loaded for the language `#1'. Using the pattern for}%
\typeout{** the default language instead.}%
\else
\language=\csname l@#1\endcsname
\fi
#2}}
\providecommand{\BIBdecl}{\relax}
\BIBdecl

\bibitem{GW98}
\BIBentryALTinterwordspacing
R.~Goodman and N.~Wallach, \emph{Representations and Invariants of the Classical Groups}, ser. Encyclopedia of Mathematics and its Applications.\hskip 1em plus 0.5em minus 0.4em\relax Cambridge University Press, 2000. [Online]. Available: \url{https://books.google.com.sg/books?id=MYFepb2yq1wC}
\BIBentrySTDinterwordspacing

\bibitem{BCH05}
\BIBentryALTinterwordspacing
D.~Bacon, I.~L. Chuang, and A.~W. Harrow, ``The quantum schur transform: I. efficient qudit circuits,'' 2006. [Online]. Available: \url{https://arxiv.org/abs/quant-ph/0601001}
\BIBentrySTDinterwordspacing

\bibitem{HarrowTh05}
\BIBentryALTinterwordspacing
A.~W. Harrow, ``Applications of coherent classical communication and the schur transform to quantum information theory,'' \emph{arXiv:quant-ph/0512255}, 2005. [Online]. Available: \url{https://arxiv.org/abs/quant-ph/0512255}
\BIBentrySTDinterwordspacing

\bibitem{KiSt18}
\BIBentryALTinterwordspacing
W.~M. Kirby and F.~W. Strauch, ``A practical quantum algorithm for the schur transform,'' \emph{Quantum Information and Computation}, vol.~18, no. 9{\&}10, pp. 721--742, aug 2018. [Online]. Available: \url{https://doi.org/10.26421\%2Fqic18.9-10-1}
\BIBentrySTDinterwordspacing

\bibitem{Krovi19}
\BIBentryALTinterwordspacing
H.~Krovi, ``An efficient high dimensional quantum {S}chur transform,'' \emph{{Quantum}}, vol.~3, p. 122, Feb. 2019. [Online]. Available: \url{https://doi.org/10.22331/q-2019-02-14-122}
\BIBentrySTDinterwordspacing

\bibitem{WS23}
A.~Wills and S.~Strelchuk, ``Generalised coupling and an elementary algorithm for the quantum schur transform,'' 2023.

\bibitem{KeyWer01}
\BIBentryALTinterwordspacing
M.~Keyl and R.~F. Werner, ``Estimating the spectrum of a density operator,'' \emph{Phys. Rev. A}, vol.~64, p. 052311, Oct 2001. [Online]. Available: \url{https://link.aps.org/doi/10.1103/PhysRevA.64.052311}
\BIBentrySTDinterwordspacing

\bibitem{ChMi06}
\BIBentryALTinterwordspacing
M.~Christandl and G.~Mitchison, ``The spectra of quantum states and the kronecker coefficients of the symmetric group,'' \emph{Communications in Mathematical Physics}, vol. 261, no.~3, pp. 789--797, Feb 2006. [Online]. Available: \url{https://doi.org/10.1007/s00220-005-1435-1}
\BIBentrySTDinterwordspacing

\bibitem{DoWr15}
\BIBentryALTinterwordspacing
R.~O'Donnell and J.~Wright, ``Quantum spectrum testing,'' in \emph{Proceedings of the Forty-Seventh Annual ACM Symposium on Theory of Computing}, ser. STOC '15.\hskip 1em plus 0.5em minus 0.4em\relax New York, NY, USA: Association for Computing Machinery, 2015, p. 529–538. [Online]. Available: \url{https://doi.org/10.1145/2746539.2746582}
\BIBentrySTDinterwordspacing

\bibitem{DoWr16}
\BIBentryALTinterwordspacing
R.~O'Donnell and J.~Wright, ``Efficient quantum tomography,'' in \emph{Proceedings of the Forty-Eighth Annual ACM Symposium on Theory of Computing}, ser. STOC '16.\hskip 1em plus 0.5em minus 0.4em\relax New York, NY, USA: Association for Computing Machinery, 2016, p. 899–912. [Online]. Available: \url{https://doi.org/10.1145/2897518.2897544}
\BIBentrySTDinterwordspacing

\bibitem{Haetal17}
J.~Haah, A.~W. Harrow, Z.~Ji, X.~Wu, and N.~Yu, ``Sample-optimal tomography of quantum states,'' \emph{IEEE Transactions on Information Theory}, vol.~63, no.~9, pp. 5628--5641, 2017.

\bibitem{HaMa03}
M.~Hayashi and K.~Matsumoto, ``Simple construction of quantum universal variable-length source coding,'' in \emph{IEEE International Symposium on Information Theory, 2003. Proceedings.}, 2003, pp. 459--.

\bibitem{Hayashi16}
\BIBentryALTinterwordspacing
Y.~Yang, G.~Chiribella, and M.~Hayashi, ``Optimal compression for identically prepared qubit states,'' \emph{Phys. Rev. Lett.}, vol. 117, p. 090502, Aug 2016. [Online]. Available: \url{https://link.aps.org/doi/10.1103/PhysRevLett.117.090502}
\BIBentrySTDinterwordspacing

\bibitem{MaHa07}
\BIBentryALTinterwordspacing
K.~Matsumoto and M.~Hayashi, ``Universal distortion-free entanglement concentration,'' \emph{Phys. Rev. A}, vol.~75, p. 062338, Jun 2007. [Online]. Available: \url{https://link.aps.org/doi/10.1103/PhysRevA.75.062338}
\BIBentrySTDinterwordspacing

\bibitem{BlCrGo14}
R.~Blume-Kohout, S.~Croke, and D.~Gottesman, ``Streaming universal distortion-free entanglement concentration,'' \emph{IEEE Transactions on Information Theory}, vol.~60, no.~1, pp. 334--350, 2014.

\bibitem{ZaRa97}
\BIBentryALTinterwordspacing
P.~Zanardi and M.~Rasetti, ``Error avoiding quantum codes,'' \emph{Modern Physics Letters B}, vol.~11, no.~25, pp. 1085--1093, 1997. [Online]. Available: \url{https://doi.org/10.1142/S0217984997001304}
\BIBentrySTDinterwordspacing

\bibitem{KnLaVi00}
\BIBentryALTinterwordspacing
E.~Knill, R.~Laflamme, and L.~Viola, ``Theory of quantum error correction for general noise,'' \emph{Phys. Rev. Lett.}, vol.~84, pp. 2525--2528, Mar 2000. [Online]. Available: \url{https://link.aps.org/doi/10.1103/PhysRevLett.84.2525}
\BIBentrySTDinterwordspacing

\bibitem{KeBaLiWha01}
\BIBentryALTinterwordspacing
J.~Kempe, D.~Bacon, D.~A. Lidar, and K.~B. Whaley, ``Theory of decoherence-free fault-tolerant universal quantum computation,'' \emph{Phys. Rev. A}, vol.~63, p. 042307, Mar 2001. [Online]. Available: \url{https://link.aps.org/doi/10.1103/PhysRevA.63.042307}
\BIBentrySTDinterwordspacing

\bibitem{BaconThesis03}
\BIBentryALTinterwordspacing
D.~Bacon, ``Decoherence, control, and symmetry in quantum computers,'' \emph{arXiv:quant-ph/0305025}, 2003. [Online]. Available: \url{https://arxiv.org/abs/quant-ph/0305025}
\BIBentrySTDinterwordspacing

\bibitem{ChiYa15}
\BIBentryALTinterwordspacing
G.~Chiribella, Y.~Yang, and C.~C. Huang, ``Universal superreplication of unitary gates,'' \emph{Phys. Rev. Lett.}, vol. 114, p. 120504, Mar 2015. [Online]. Available: \url{https://link.aps.org/doi/10.1103/PhysRevLett.114.120504}
\BIBentrySTDinterwordspacing

\bibitem{ChiYa16}
\BIBentryALTinterwordspacing
G.~Chiribella and Y.~Yang, ``Quantum superreplication of states and gates,'' \emph{Frontiers of Physics}, vol.~11, no.~3, p. 110304, Mar 2016. [Online]. Available: \url{https://doi.org/10.1007/s11467-016-0556-7}
\BIBentrySTDinterwordspacing

\bibitem{KoeMit09}
\BIBentryALTinterwordspacing
R.~Koenig and G.~Mitchison, ``A most compendious and facile quantum de finetti theorem,'' \emph{Journal of Mathematical Physics}, vol.~50, no.~1, 2009. [Online]. Available: \url{https://www.osti.gov/biblio/21175884}
\BIBentrySTDinterwordspacing

\bibitem{Gross21}
\BIBentryALTinterwordspacing
D.~Gross, S.~Nezami, and M.~Walter, ``Schur--weyl duality for the clifford group with applications: Property testing, a robust hudson theorem, and de finetti representations,'' \emph{Communications in Mathematical Physics}, vol. 385, no.~3, pp. 1325--1393, Aug 2021. [Online]. Available: \url{https://doi.org/10.1007/s00220-021-04118-7}
\BIBentrySTDinterwordspacing

\bibitem{Ragoneetal22}
\BIBentryALTinterwordspacing
M.~Ragone, P.~Braccia, Q.~T. Nguyen, L.~Schatzki, P.~J. Coles, F.~Sauvage, M.~Larocca, and M.~Cerezo, ``Representation theory for geometric quantum machine learning,'' \emph{arXiv:2210.07980}, 2022. [Online]. Available: \url{https://arxiv.org/abs/2210.07980}
\BIBentrySTDinterwordspacing

\bibitem{Nguyenetal22}
\BIBentryALTinterwordspacing
Q.~T. Nguyen, L.~Schatzki, P.~Braccia, M.~Ragone, P.~J. Coles, F.~Sauvage, M.~Larocca, and M.~Cerezo, ``Theory for equivariant quantum neural networks,'' \emph{arXiv:2210.08566}, 2022. [Online]. Available: \url{https://arxiv.org/abs/2210.08566}
\BIBentrySTDinterwordspacing

\bibitem{Schatzkietal22}
\BIBentryALTinterwordspacing
L.~Schatzki, M.~Larocca, Q.~T. Nguyen, F.~Sauvage, and M.~Cerezo, ``Theoretical guarantees for permutation-equivariant quantum neural networks,'' \emph{arXiv:2210.09974}, 2022. [Online]. Available: \url{https://arxiv.org/abs/2210.09974}
\BIBentrySTDinterwordspacing

\bibitem{Buhrmanetal22}
\BIBentryALTinterwordspacing
H.~Buhrman, N.~Linden, L.~Mančinska, A.~Montanaro, and M.~Ozols, ``Quantum majority vote,'' \emph{arXiv:2211.11729}, 2022. [Online]. Available: \url{https://arxiv.org/abs/2211.11729}
\BIBentrySTDinterwordspacing

\bibitem{Nguyen_2023}
Q.~T. Nguyen, ``The mixed schur transform: efficient quantum circuit and applications,'' 2023.

\bibitem{Grinko_2023}
D.~Grinko, A.~Burchardt, and M.~Ozols, ``Gelfand-tsetlin basis for partially transposed permutations, with applications to quantum information,'' 2023.

\bibitem{GO23}
D.~Grinko and M.~Ozols, ``Linear programming with unitary-equivariant constraints,'' 2023.

\bibitem{IH08}
\BIBentryALTinterwordspacing
S.~Ishizaka and T.~Hiroshima, ``Asymptotic teleportation scheme as a universal programmable quantum processor,'' \emph{Phys. Rev. Lett.}, vol. 101, p. 240501, Dec 2008. [Online]. Available: \url{https://link.aps.org/doi/10.1103/PhysRevLett.101.240501}
\BIBentrySTDinterwordspacing

\bibitem{KM+21}
\BIBentryALTinterwordspacing
P.~Kopszak, M.~Mozrzymas, M.~Studzi{\'{n}}ski, and M.~Horodecki, ``Multiport based teleportation – transmission of a large amount of quantum information,'' \emph{{Quantum}}, vol.~5, p. 576, Nov. 2021. [Online]. Available: \url{https://doi.org/10.22331/q-2021-11-11-576}
\BIBentrySTDinterwordspacing

\bibitem{NPR21}
\BIBentryALTinterwordspacing
I.~Nechita, C.~Pellegrini, and D.~Rochette, ``A geometrical description of the universal {$1 \rightarrow 2$} asymmetric quantum cloning region,'' \emph{Quantum Information Processing}, vol.~20, no.~10, p. 333, Oct 2021. [Online]. Available: \url{https://doi.org/10.1007/s11128-021-03258-y}
\BIBentrySTDinterwordspacing

\bibitem{KL22}
\BIBentryALTinterwordspacing
L.~Kong and Z.-W. Liu, ``Near-optimal covariant quantum error-correcting codes from random unitaries with symmetries,'' \emph{PRX Quantum}, vol.~3, p. 020314, Apr 2022. [Online]. Available: \url{https://link.aps.org/doi/10.1103/PRXQuantum.3.020314}
\BIBentrySTDinterwordspacing

\bibitem{HKMV22}
\BIBentryALTinterwordspacing
F.~Huber, I.~Klep, V.~Magron, and J.~Vol{\v{c}}i{\v{c}}, ``Dimension-free entanglement detection in multipartite werner states,'' \emph{Communications in Mathematical Physics}, vol. 396, no.~3, pp. 1051--1070, Dec 2022. [Online]. Available: \url{https://doi.org/10.1007/s00220-022-04485-9}
\BIBentrySTDinterwordspacing

\bibitem{YSM23}
\BIBentryALTinterwordspacing
S.~Yoshida, A.~Soeda, and M.~Murao, ``Reversing unknown qubit-unitary operation, deterministically and exactly,'' \emph{Physical Review Letters}, vol. 131, no.~12, Sep. 2023. [Online]. Available: \url{http://dx.doi.org/10.1103/PhysRevLett.131.120602}
\BIBentrySTDinterwordspacing

\bibitem{Cirac_99}
\BIBentryALTinterwordspacing
J.~I. Cirac, A.~K. Ekert, and C.~Macchiavello, ``Optimal purification of single qubits,'' \emph{Physical Review Letters}, vol.~82, no.~21, p. 4344–4347, May 1999. [Online]. Available: \url{http://dx.doi.org/10.1103/PhysRevLett.82.4344}
\BIBentrySTDinterwordspacing

\bibitem{Werner_1998}
\BIBentryALTinterwordspacing
R.~F. Werner, ``Optimal cloning of pure states,'' \emph{Physical Review A}, vol.~58, no.~3, p. 1827–1832, Sep. 1998. [Online]. Available: \url{http://dx.doi.org/10.1103/PhysRevA.58.1827}
\BIBentrySTDinterwordspacing

\bibitem{Cervero_2024}
\BIBentryALTinterwordspacing
Y.~Hu, E.~Cervero-Martín, E.~Theil, L.~Mančinska, and M.~Tomamichel, ``Sample optimal and memory efficient quantum state tomography,'' 2024. [Online]. Available: \url{https://arxiv.org/abs/2410.16220}
\BIBentrySTDinterwordspacing

\bibitem{Kuperberg23}
G.~Kuperberg, ``Breaking the cubic barrier in the solovay-kitaev algorithm,'' 2023.

\bibitem{Bacon_2006}
\BIBentryALTinterwordspacing
D.~Bacon, I.~L. Chuang, and A.~W. Harrow, ``Efficient quantum circuits for schur and clebsch-gordan transforms,'' \emph{Physical Review Letters}, vol.~97, no.~17, Oct. 2006. [Online]. Available: \url{http://dx.doi.org/10.1103/PhysRevLett.97.170502}
\BIBentrySTDinterwordspacing

\bibitem{Koike_89}
\BIBentryALTinterwordspacing
K.~Koike, ``On the decomposition of tensor products of the representations of the classical groups: By means of the universal characters,'' \emph{Advances in Mathematics}, vol.~74, no.~1, pp. 57--86, 1989. [Online]. Available: \url{https://www.sciencedirect.com/science/article/pii/0001870889900042}
\BIBentrySTDinterwordspacing

\bibitem{CM23}
\BIBentryALTinterwordspacing
E.~Cervero and L.~Mančinska, ``Weak schur sampling with logarithmic quantum memory,'' 2023. [Online]. Available: \url{https://arxiv.org/abs/2309.11947}
\BIBentrySTDinterwordspacing

\bibitem{Vilenkin1995}
\BIBentryALTinterwordspacing
N.~J. Vilenkin and A.~U. Klimyk, \emph{Representation of Lie Groups and Special Functions}.\hskip 1em plus 0.5em minus 0.4em\relax Springer Netherlands, 1995. [Online]. Available: \url{http://dx.doi.org/10.1007/978-94-017-2885-0}
\BIBentrySTDinterwordspacing

\bibitem{Nielsen_Chuang_2010}
M.~A. Nielsen and I.~L. Chuang, \emph{Quantum Computation and Quantum Information: 10th Anniversary Edition}.\hskip 1em plus 0.5em minus 0.4em\relax Cambridge University Press, 2010.

\bibitem{Sagan_2011}
B.~E. Sagan, \emph{The symmetric group - representations, combinatorial algorithms, and symmetric functions.}, ser. Wadsworth \& Brooks / Cole mathematics series.\hskip 1em plus 0.5em minus 0.4em\relax Wadsworth, 1991.

\bibitem{Goodman2009}
\BIBentryALTinterwordspacing
R.~Goodman and N.~R. Wallach, \emph{Symmetry, Representations, and Invariants}.\hskip 1em plus 0.5em minus 0.4em\relax Springer New York, 2009. [Online]. Available: \url{http://dx.doi.org/10.1007/978-0-387-79852-3}
\BIBentrySTDinterwordspacing

\end{thebibliography}
